\documentclass[reqno]{amsart}

\usepackage{amsmath, amssymb, amsfonts, amsthm}

\theoremstyle{plain}
\newtheorem{theorem}{Theorem}[section] 

\newtheorem{lemma}[theorem]{Lemma}

\newtheorem{definition}[theorem]{Definition}

\theoremstyle{remark}
\newtheorem{remark}{Remark}[section]

\usepackage{bm}

\newcommand\Pf{\mathrm{Pf}}

\begin{document}

\title{\mbox{Planar Ising model at criticality:} \mbox{state-of-the-art and perspectives}}

\author{Dmitry Chelkak}

\address{Holder of the ENS--MHI chair funded by MHI. \newline
\indent D\'epartement de math\'ematiques et applications de l'ENS, \newline
\indent \'Ecole Normale Sup\'erieure PSL Research University, CNRS UMR 8553, Paris 5\`eme. \newline
\indent On leave from St.~Petersburg Department of Steklov Mathematical Institute RAS.}

\email{dmitry.chelkak@ens.fr}

\begin{abstract} In this essay, we briefly discuss recent developments, started a decade ago in the seminal work of Smirnov and continued by a number of authors, centered around the conformal invariance of the critical planar Ising model on~$\mathbb{Z}^2$ and, more generally, of the critical Z-invariant Ising model on isoradial graphs (rhombic lattices). We also introduce a new class of embeddings of general weighted planar graphs ({s-embeddings}), which might, in particular, pave the way to true universality results for the planar Ising model.
\end{abstract}



\maketitle

\section{Introduction}

The two-dimensional Ising model, introduced by Lenz almost a hundred years ago, does not need an introduction, being probably the most famous example of a statistical mechanics system exhibiting the phase transition and the conformally invariant behavior at criticality, as well as an inspiring structure of the correlation functions both at the critical point and in a vicinity of it; e.g., see~\cite{mccoy-wu-book} and~\cite{mussardo-book}. More recently, it became a playground for mathematicians interested in a rigorous understanding of the conjectural conformal invariance of critical lattice systems~\cite{smirnov-06icm}.

What makes the \emph{planar} (a priori, not necessarily critical) Ising model particularly feasible for such a mathematical treatment (in absence of the magnetic field) is the underlying structure of \emph{\mbox{s-holo}\-morphic spinors} (aka \emph{fermionic observables}), essentially dating back to the work of Onsager and Kaufman and reinterpreted several times since then, notably by Kadanoff and Ceva~\cite{kadanoff-ceva}. From the classical analysis (or probability theory) perspective, these \mbox{s-holo}morphic spinors can be thought of as counterparts of discrete harmonic functions associated to random-walk-based systems. The main theme of this note is recent convergence results for the critical model based on an analysis of such observables in the small mesh size limit.

The text should \emph{not} be considered as a survey: certainly, many results and references deserving to be mentioned in such a survey are missing. Its purposes are rather to give a general informal description of the current state-of-the-art of the subject (from the personal viewpoint of the author and to the best of his knowledge) for readers not interested in technical details; to provide some references for those interested; and to indicate several ongoing research directions and open questions.

I wish to thank my co-authors as well as many other colleagues for numerous helpful discussions (centered at Geneva a decade ago, worldwide nowadays) of the planar Ising model and for sharing their knowledge and ideas with others. Clearly, the progress described below was achieved by collective efforts and it is a privilege to discuss the results of hard work of a rather big community in this \emph{essay}.

\newpage

\section{Discrete spinors and s-holomorphicity in the planar Ising model}
\label{sec:formalism}

\subsection{Notation and the Kramers--Wannier duality}

Below we consider the ferromagnetic Ising model on \emph{faces} of a graph~$G$ embedded into the complex plane~$\mathbb{C}$ so that all its edges are straight segments. One can also work with graphs embedded into Riemann surfaces but for simplicity we prefer not to discuss such a generalization here (see~\cite[Section~4]{chelkak-cimasoni-kassel} and references therein for more details). A \emph{spin configuration}~$\sigma$ is an assignment of a~$\pm 1$ spin~$\sigma_u$ to each face of~$G$, including the outer face~$u_\mathrm{out}$, with the spin~$\sigma_{u_{\mathrm{out}}}$ playing the role of boundary conditions. The probability measure~$\mathbb{P}^\circ$ (the superscipt~$\circ$ emphasizes the fact that the model is considered on faces of~$G$) on the set of spin configurations is given by
\begin{equation}
\label{eq:Ising_proba}
\textstyle \mathbb{P}^\circ(\sigma)=(\mathcal{Z}^\circ(G))^{-1}\exp\big[\beta\sum_{e} J_{e}\sigma_{u_-(e)}\sigma_{u_+(e)}\big],
\end{equation}
where~$\beta>0$ is the inverse temperature, $J_e>0$ are fixed interaction constants, the summation is over unoriented edges~of~$G$ (an edge~$e$ separates two faces~$u_\pm(e)$), and~$\mathcal{Z}^\circ(G)$ is the normalization constant called the \emph{partition function}. Note that the spin of the outer face~$u_\mathrm{out}$ of~$G$ plays the role of boundary conditions and one can always break the~$\mathbb{Z}_2$ (spin-flip) symmetry of the model by assuming~$\sigma_{u_\mathrm{out}}=+1$.

Abusing the notation slightly, we also admit the situation when~$J_e=0$ along some boundary arcs of~$G$, which means that the corresponding near-to-boundary spins do not interact with~$\sigma_{u_\mathrm{out}}$. We call these parts of the boundary of~$G$ \emph{free arcs} and use the name \emph{wired arcs} for the remaining parts, across which the inner spins interact with the \emph{same}~$\sigma_{u_{\mathrm{out}}}=+1$. We will use the name \emph{standard boundary conditions} for this setup, see also~\cite{chelkak-hongler-izyurov-mixed}. For simplicity, below we always assume that there exists at least one wired arc.

We denote by~$G^\bullet$ the graph obtained from~$G$ by removing the edges along free parts of the boundary and by~$G^\circ$ a graph dual to~$G$ with the following convention: instead of a single vertex corresponding to~$u_\mathrm{out}$ we draw one vertex per boundary edge on wired arcs. Combinatorially, \emph{all} these vertices of~$G^\circ$ should be thought of as wired together (hence the name) and, similarly, the vertices of~$G^\bullet$ along free parts of the boundary should be thought of as a collection of `macro-vertices', one per free arc. We assume that~$G^\circ$ is also embedded into~$\mathbb{C}$ so that all its edges are straight segments and denote by~$\diamondsuit(G)$ the set of quads~$(vuv'u')$ formed by pairs of (unoriented) dual edges~$(vv')$ and~$(uu')$ of~$G^\bullet$ and~$G^\circ$, respectively. We also denote by~$\partial\diamondsuit(G)$ the set of triangles $(vuv')$ and~$(uvu')$ arising instead of quads along free and wired boundary arcs, respectively, and set~$\overline{\diamondsuit(G)}:=\diamondsuit(G)\cup \partial\diamondsuit(G)$.
Finally, let~$\Omega(G)\subset \mathbb{C}$ be the polygon formed by all these quads and triangles.

For an unoriented edge~$e$ of~$G^\bullet$ (or, equivalently, an element of~$\diamondsuit(G)$), we define~$x_e:=\exp[-2\beta J_e]$ and extend this notation to the elements of~$\partial\diamondsuit(G)$ by setting~$x_e:=1$ on free arcs and~$x_e:=0$ on wired ones. For a subset~$C\subset\diamondsuit(G)$, denote~$x(C):=\prod_{e\in C}x_e$ and let~$\mathcal{E}(G)$ denote the set of all even subgraphs of~$G$. There exists a trivial bijection of this set and the set of spin configurations on faces of~$G$: draw edges separating misaligned spins. In particular, one sees that
\begin{equation}
\label{eq:Z_domain_walls}
\textstyle \mathcal{Z}^\circ(G)=\prod_{e\in\diamondsuit(G)} x_e^{-1/2}\cdot \mathcal{Z}(G),\quad \text{where}\quad \mathcal{Z}(G) := \sum_{C\in\mathcal{E}(G)}x(C),
\end{equation}
this is called the \emph{domain walls} (or low-temperature) \emph{expansion} of~$\mathcal{Z}^\circ(G)$. A remarkable fact (first observed by van der Waerden) is that the \emph{same}~$\mathcal{Z}(G)$ also gives an expression for the partition function~$\mathcal{Z}^\bullet(G)$ of the Ising model on~\emph{vertices} of~$G^\bullet$, provided that the dual parameters~$\beta^\bullet$ and~$J^\bullet_e$ satisfy the identity~$x_e=\tanh[\beta^\bullet J^\bullet_e]$. Namely, the following \emph{high-temperature expansion} of~$\mathcal{Z}^\bullet(G)$ holds true:
\begin{equation}
\label{eq:Z_high_temp}
\textstyle \mathcal{Z}^\bullet(G) = 
2^{|V(G^\bullet)|}\prod_{e\in\diamondsuit(G)}(1\!-\!x_e^2)^{-1/2}\cdot \mathcal{Z}(G),\qquad \mathcal{Z}(G)= \sum_{C\in\mathcal{E}(G)}x(C).
\end{equation}

This link between the Ising models on~$G^\circ$ and~$G^\bullet$ is called the \emph{Kramers--Wannier duality} and it is not limited to the equality between the partition functions~$\mathcal{Z^\circ}(G)$ and~$\mathcal{Z}^\bullet(G)$. Nevertheless, it is worth mentioning that similar objects typically lead to different types of sums in the two representations. E.g., in order to compute the expectation~$\mathbb{E}^\circ[\sigma_{u}\sigma_{u'}]$ similarly to~\eqref{eq:Z_domain_walls} one should keep track of the parity of the number of loops in~$C$ separating~$u$ and~$u'$, while computing the expectation~$\mathbb{E}^\bullet[\sigma_v\sigma_{v'}]$ amounts to the replacement of~$\mathcal{E}(G)$ in~\eqref{eq:Z_high_temp} by the set~$\mathcal{E}(G;v,v')$ of subgraphs of~$G$ in which~$v$ and~$v'$ (but no other vertex) have odd degrees. Also, note that the boundary conditions on~$G^\circ$ and~$G^\bullet$ are \emph{not} fully symmetric: dual spins on different free arcs are not required to coincide, contrary to the wired ones.

It is convenient to introduce the following parametrization~$\theta_e$ of the weights~$x_e$:
\begin{equation}
\label{eq:theta_e_def}
x_e=\exp[-2\beta J_e]=\tanh[\beta^\bullet J^\bullet_e]~=~\tan \tfrac{1}{2}\theta_e\,, \quad \theta_e\in [0,\tfrac{\pi}{2}].
\end{equation}
Note that, if we similarly define~$\theta^\bullet_e$ so that~$\tan\tfrac{1}{2}\theta^\bullet_e = \exp[-2\beta^\bullet J^\bullet_e]$, then~$\theta^\bullet_e=\frac{\pi}{2}-\theta_e$.

\subsection{Spins-disorders formalism of Kadanoff and Ceva}
Following~\cite{kadanoff-ceva}, we now describe the so-called \emph{disorder insertions} -- the objects dual to spins under the Kramers--Wannier duality. We also refer the interested reader to recent papers of Dub\'edat~\cite{dubedat-bosonization,dubedat-abelian}. Let~$m$ be even and~$v_1,...,v_m$ be a collection of vertices of~$G^\bullet$ (with the standard convention that each free arc should be considered as a single vertex). Let us fix a collection of paths~$\varkappa^{[v_1,...,v_m]}$ on~$G$ linking these vertices pairwise and change the interaction constants~$J_e$ to~$-J_e$ along these paths to get another probability measure on the spin configurations instead of~\eqref{eq:Ising_proba}. Note that one can think of this operation as putting an additional random weight~$\exp[-2\beta J_{uu'}\sigma_{u}\sigma_{u'}]$ along~$\gamma^{[v_1,...,v_m]}$ and treat this weight as a random variable, which we denote by~$\mu_{v_1}...\mu_{v_m}$ (note that its definition implicitly depends on the choice of disorder lines). The domain walls representation of the Ising model on~$G^\circ$ then gives
\begin{equation}
\label{eq:mu_domain_walls}
\textstyle \mathbb{E}^\circ[\mu_{v_1}...\mu_{v_m}]= x(\gamma^{[v_1,...,v_m]})\cdot \sum_{C\in\mathcal{E}(G)}x^{[v_1,...,v_m]}(C)\,/\, \mathcal{Z}(G),
\end{equation}
where the weights~$x^{[v_1,...,v_m]}$ are obtained from~$x$ by changing~$x_e$ to~$x_e^{-1}$ on disorder lines and the first factor comes from the prefactor in~\eqref{eq:Z_domain_walls}. More invariantly, one can consider the sign-flip-symmetric Ising model on \emph{faces} of the double-cover~$G^{[v_1,...,v_m]}$ of~$G$ ramified over the vertices~$v_1,...,v_m$ (the spins at two faces of~$G^{[v_1,...,v_m]}$ lying over the same face of~$G$ are required to have opposite values) and rewrite~\eqref{eq:mu_domain_walls} as
\begin{align}
\notag
\textstyle \mathbb{E}^\circ[\mu_{v_1}...\mu_{v_m}] & = \mathcal{Z}^\circ(G^{[v_1,...,v_m]})\,/\,\mathcal{Z}^\circ(G) \\
& \label{eq:mu_circ=sigma_bullet}
\textstyle  = \sum_{C\in\mathcal{E}(G;v_1,...,v_m)}x(C)\,/\,\mathcal{Z}(G)= \mathbb{E}^\bullet[\sigma_{v_1}...\sigma_{v_m}],
\end{align}
where~$\mathcal{E}(G;v_1,...,v_m)$ stands for the set of subgraphs of~$G$ in which all~$v_1,...,v_m$ (but no other vertex) have odd degrees; the last equality is the classical high-temperature expansion of spin correlations on~$G^\bullet$ mentioned above.

Vice versa, given an even~$n$ and a collection of faces~$u_1,...,u_n$ of~$G$, one can write
\begin{align}
\notag \mathbb{E}^\circ[\sigma_{u_1}...\sigma_{u_n}]&\textstyle =\sum_{C\in\mathcal{E}(G)}x_{[u_1,...,u_n]}(C)\,/\,\mathcal{Z}(G)\\
& \label{eq:sigma_circ=mu_bullet}
=\mathcal{Z}^\bullet(G_{[u_1,...,u_n]})\,/\, \mathcal{Z}^\bullet(G)=\mathbb{E}^\bullet[\mu_{u_1}...\mu_{u_n}],
\end{align}
where the weights $x_{[u_1,...,u_n]}$ are obtained from~$x$ by putting additional minus signs on the edges of~$\gamma_{[u_1,...,u_n]}$ and~$\mathcal{Z}^\bullet(G_{[u_1,...,u_n]})$ denotes the partition function of the spin-flip symmetric Ising model on \emph{vertices} of the double-cover~$G_{[u_1,...,u_n]}$ of~$G$ ramified over faces~$u_1,...,u_m$, treated via the high-temperature expansion.
Generalizing~\eqref{eq:mu_circ=sigma_bullet} and~\eqref{eq:sigma_circ=mu_bullet}, one has the following duality between spins and disorders:
\begin{align}
\label{eq:mu_sigma_circ=sigma_mu_bullet}
\mathbb{E}^\circ[\mu_{v_1}...\mu_{v_m}\sigma_{u_1}...\sigma_{u_n}] = \mathbb{E}^\bullet[\sigma_{v_1}...\sigma_{v_m}\mu_{u_1}...\mu_{u_n}]
\end{align}
since both sides are equal to~$\textstyle \sum_{C\in\mathcal{E}(G;v_1,...,v_n)}x_{[u_1,...,u_n]}(C)\,/\,\mathcal{Z}(G)$.
Let us emphasize that one needs to fix disorder lines in order to interpret these quantities as expectations with respect to~$\mathbb{P}^\circ$ and~$\mathbb{P}^\bullet$, respectively. Below we prefer a more invariant approach and view~$u_q$'s as faces of the double-cover~$G^{[v_1,...,v_n]}$ and the expectation taken with respect to the sign-flip symmetric Ising model defined on faces of this double-cover. To avoid possible confusion, we introduce the notation
\begin{equation}
\label{eq:corr_mu_sigma_def}
\langle\mu_{v_1}...\mu_{v_m}\sigma_{u_1}...\sigma_{u_n}\rangle:= \mathbb{E}^\circ_{G^{[v_1,...,v_m]}}[\sigma_{u_1}...\sigma_{u_n}]
\end{equation}
instead of~\eqref{eq:mu_sigma_circ=sigma_mu_bullet}. Considered as a function of both~$v_p$'s and~$u_q$'s, \eqref{eq:corr_mu_sigma_def} is defined on a double-cover~$G^{m,n}_{[\bullet,\circ]}$ of~$(G^\bullet)^m\times(G^\circ)^n$ and changes the sign each time when one of~$v_p$ turns around one of~$u_q$ or vice versa; we call such functions \emph{spinors} on~$G^{m,n}_{[\bullet,\circ]}$.

We also need some additional notation. Let~$\Lambda(G):=G^\bullet\cup G^\circ$ be the planar graph formed by the sides of quads from~$\diamondsuit(G)$ and let~$\Upsilon(G)$ denote the medial graph of~$\Lambda(G)$. In other words, the vertices of~$\Upsilon(G)$ correspond to edges~$(uv)$ of~$\Lambda(G)$ or to \emph{corners} of~$G$, while the faces of~$\Upsilon(G)$ correspond either to vertices of~$G^\bullet$ or to vertices of~$G^\circ$ or to quads from~$\diamondsuit(G)$. Denote by~$\Upsilon^\times(G)$ a double-cover of~$\Upsilon(G)$ which branches over each of its faces (e.g., see~\cite[Fig.~27]{mercat-CMP} or~\cite[Fig.~6]{chelkak-smirnov}). Note that~$\Upsilon^\times(G)$ is fully defined by this condition for graphs embedded into~$\mathbb{C}$ but on Riemann surfaces there is a choice that can be rephrased as the choice of a \emph{spin structure} on the surface. Below we discuss spinors on~$\Upsilon^\times(G)$, i.e. the functions whose values at two vertices of~$\Upsilon^\times(G)$ lying over the same vertex of~$\Upsilon(G)$ differ by the sign. An important example of such a function is the \emph{Dirac spinor}
\begin{equation}
\label{eq:Dirac_spinor}
\textstyle \eta_c:=\varsigma\cdot\exp[-\frac{i}{2}\arg(v(c)-u(c))],\quad \text{where}\quad c=(u(c)v(c))\in \Upsilon^\times(G),
\end{equation}
$u(c)\in G^\circ$,~$v(c)\in G^\bullet$ and a global prefactor~$\varsigma:|\varsigma|\!=\!1$ is added to the definition for later convenience. If~$G$ was embedded into a Riemann surface~$\Sigma$, one should fix a vector field on~$\Sigma$ with zeroes of even index to define the~$\arg$ function (cf.~\cite[Section~4]{chelkak-cimasoni-kassel}), in the case~$\Sigma=\mathbb{C}$ we simply consider a constant vector field~$\varsigma^2$.

Given a corner~$c$ of~$G$, we formally define~$\chi_c:=\mu_{v(c)}\sigma_{u(c)}$, i.e.
\begin{equation}
\label{eq:chi:=mu_sigma_def}
\langle\chi_c\mu_{v_1}...\mu_{v_{m-1}}\sigma_{u_1}...\sigma_{u_{n-1}}\rangle:= \langle\mu_{v(c)}\mu_{v_1}...\mu_{v_{m-1}}\sigma_{u(c)}\sigma_{u_1}...\sigma_{u_{n-1}}\rangle.
\end{equation}
According to the preceding discussion of mixed spins-disorders expectations, for a given collection~$\varpi:=\{v_1,...,v_{m-1},u_1,...,u_{n-1}\}\in (G^\bullet)^{m-1}\times(G^\circ)^{n-1}$, the function~\eqref{eq:chi:=mu_sigma_def} locally behaves as a spinor on~$\Upsilon^\times(G)$ but its global branching structure is slightly different as the additional sign change arises when~$c$ turns around one of~$v_p$ or~$u_q$. Let us denote the corresponding double-cover of~$\Upsilon(G)$ by~$\Upsilon^\times_\varpi(G)$. Finally, define~$\psi_c:=\eta_c\chi_c$, where~$\eta_c$ is defined by~\eqref{eq:Dirac_spinor}, and note that functions
\begin{equation}
\label{eq:psi:=eta_chi_def}
\langle\psi_c\mu_{v_1}...\mu_{v_{m-1}}\sigma_{u_1}...\sigma_{u_{n-1}}\rangle:= \eta_c\langle\chi_c\mu_{v_1}...\mu_{v_{m-1}}\sigma_{u_1}...\sigma_{u_{n-1}}\rangle
\end{equation}
do not branch locally (because of the cancellation of sign changes of~$\chi_c$ and~$\eta_c$) and change the sign only when~$c$ turns around one of the vertices $v_p$ or the faces~$u_q$. We denote the corresponding (i.e., ramified over~$\varpi$) double-cover of~$\Upsilon(G)$ by~$\Upsilon_\varpi(G)$.

\subsection{S-holomorphicity} \label{sub:s-hol}
This section is devoted to the crucial three-term equation for the functions~\eqref{eq:chi:=mu_sigma_def}, the so-called \emph{propagation equation} for spinors on~$\Upsilon^\times(G)$ (and also known as the Dotsenko--Dotsenko equation, see~\cite{mercat-CMP} and~\cite[Section~3.5]{chelkak-cimasoni-kassel}). To simplify the presentation, we introduce the following notation: for a quad $z_e$ corresponding to an edge~$e$ of~$G$, we denote its vertices, \emph{listed counterclockwise}, by~$v^\bullet_0(z_e)$, $v^\circ_0(z_e)$, $v^\bullet_1(z_e)$, and~$v^\circ_1(z_e)$, where~$v^\bullet_0(z),v^\bullet_1(z)\in G^\bullet$ and~$v^\circ_0(z),v^\circ_1(z)\in G^\circ$ (the choice of~$v^\bullet_0(z)$ among the two vertices of~$G^\bullet$ is arbitrary). The corner of~$G$ corresponding to the edge~$(v^\bullet_p(z_e)v^\circ_q(z_e))$ of~$z_e$ is denoted by~$c_{pq}(z_e)\in\Upsilon(G)$. For shortness, we also often omit~$z_e$ in this notation if no confusion arises.

\begin{definition}
\label{def:s-hol}
A spinor~$F$ defined on~$\Upsilon^\times(G)$ or, more generally, on some~$\Upsilon^\times_\varpi(G)$ is called \emph{s-holomorphic} if its values at any three consecutive (on~$\Upsilon^\times_\varpi(G)$) corners $c_{p,1-q}(z_e)$, $c_{pq}(z_e)$ and $c_{1-p,q}(z_e)$ surrounding a quad~$z_e\in\diamondsuit(G)$ satisfy the identity
\begin{equation}
\label{eq:propagation_on_corners}
F(c_{pq})=F(c_{p,1-q})\cos\theta_e+F(c_{1-p,q})\sin\theta_e\,,
\end{equation}
where~$\theta_e$ stands for the parametrization~\eqref{eq:theta_e_def} of the Ising model weight~$x_e$ of~$e$.
\end{definition}

\begin{remark} In fact, a straightforward computation shows that~\eqref{eq:propagation_on_corners} implies the spinor property:~$F(c^\flat_{pq})=-F(c^\sharp_{pq})$ if~$c^\flat_{pq},c^\sharp_{pq}\in\Upsilon^\times_\varpi(G)$ lie over the same corner~$c_{pq}$.
\end{remark}

The key observation is that all the Ising model observables of the form~\eqref{eq:chi:=mu_sigma_def}, considered as functions of~$c\in\Upsilon^\times_\varpi(G)$, satisfy the propagation equation~\eqref{eq:propagation_on_corners} (e.g., see see~\cite[Section~3.5]{chelkak-cimasoni-kassel}). In the recent research, this equation was mostly used in the context of isoradial graphs, in which case the parameter~$\theta_e$ has also a direct geometric meaning, but in fact~\eqref{eq:propagation_on_corners} is fairly abstract. In particular, Definition~\ref{def:s-hol} does \emph{not} rely upon a particular choice (up to a homotopy) of an embedding of~$\diamondsuit(G)$ into~$\mathbb{C}$. Contrary to~\eqref{eq:propagation_on_corners}, which was known for decades, the next definition first appeared in the work of Smirnov~\cite{smirnov-06icm,smirnov-conf-inv-I} on the critical Ising model on~$\mathbb{Z}^2$.

\begin{definition}
\label{def:H_F}
Let~$F$ be an s-holomorphic spinor on~$\Upsilon^\times_\varpi(G)$. Then one can define a function~$H_F$ on~$\Lambda(G)$ by specifying its increment around each quad from~$\diamondsuit(G)$ as
\begin{equation}
\label{eq:def_of_H}
H_F(v_p^\bullet)-H_F(v_q^\circ)~:=~(F(c^\sharp_{pq}))^2=(F(c^\flat_{pq}))^2
\end{equation}
Note that, independently of~$\varpi$, $H_F$ is defined on~$\Lambda(G)$ and not on its double-cover.
\end{definition}
\begin{remark}
Due to~\eqref{eq:propagation_on_corners}, one has~$(F(c_{00}))^2\!+(F(c_{11}))^2\!=(F(c_{01}))^2\!+(F(c_{10}))^2$. Thus, $H_F$ is locally (and hence globally in the simply connected setup) well-defined.
\end{remark}
A priori, the functions~$H_F$ do not seem to be natural objects for the study of correlations in the Ising model but they turned out to be absolutely indispensable for the analysis of scaling limits of such correlations in discrete approximations to general planar domains initiated in~\cite{smirnov-06icm,smirnov-conf-inv-I}, see Section~\ref{sub:boudnary_H} below.

\subsection{Pfaffian structure of fermionic correlators} Similarly to~\eqref{eq:chi:=mu_sigma_def}, one can consider expectations containing two or more formal variables~$\chi_c$. We start with the following observation: despite the fact that the quantities~\eqref{eq:corr_mu_sigma_def}, viewed as functions on~$G^{m,n}_{[\bullet,\circ]}$, are symmetric with respect to permutations of~$v_p$, as well as to those of~$u_q$, an additional sign change appears if one exchanges~$(u(c)v(c))$ and~$(u(d)v(d))$; this can be viewed as a cumulative result of a `half-turn' of~$u(c)$ around~$v(d)$ and a `half-turn' of~$u(d)$ around~$v(c)$. In other words, the variables~$\chi_c$ and~$\chi_d$ \emph{anti-commute}:
\begin{equation}
\langle\chi_d\chi_c\mu_{v_1}...\mu_{v_{m-2}}\sigma_{u_1}...\sigma_{u_{n-2}}\rangle= -\langle\chi_c\chi_d\mu_{v_1}...\mu_{v_{m-2}}\sigma_{u_1}...\sigma_{u_{n-2}}\rangle
\end{equation}
if one considers both sides as a function of~$(c,d)\in (\Upsilon^\times_\varpi(G))^{[2]}$, where $(\Upsilon^\times_\varpi(G))^{[2]}$ denotes the set $(\Upsilon^\times_\varpi(G))^2\setminus\{(c,d):c^{\sharp,\flat}= d^{\sharp,\flat}~\text{or}~c^{\sharp,\flat}=d^{\flat,\sharp}\}$. More generally, given a collection of vertices and faces~$\varpi=\{v_1,...,v_{m-k},u_1,...,u_{n-k}\}$, the quantities
\begin{equation}
\label{eq:many_chi}
\langle\chi_{c_1}...\chi_{c_k}\mathcal{O}_\varpi[\mu,\sigma]\,\rangle~:=~
\langle\chi_{c_1}...\chi_{c_k}\mu_{v_1}...\mu_{v_{m-k}}\sigma_{u_1}...\sigma_{u_{n-k}}\rangle
\end{equation}
are anti-symmetric functions on~$(\Upsilon^\times_\varpi(G))^{[k]}$; see~\cite{chelkak-hongler-izyurov-mixed} for more precise definitions.

Another striking observation, which is also well known in the folklore for decades, is that the correlations~\eqref{eq:many_chi}
satisfy the Pfaffian identities: for an even number of pairwise distinct~$c_1,...,c_k$ corners of~$G$, one has
\begin{align}
\langle\chi_{c_1}...\chi_{c_k}\mathcal{O}_\varpi[\mu,\sigma]\,\rangle \cdot \langle\mathcal{O}_\varpi[\mu,\sigma]\,\rangle^{k/2-1}
 = \label{eq:Pfaff_chi} ~\Pf\big[\,\langle\chi_{c_r}\chi_{c_s}\mathcal{O}_\varpi[\mu,\sigma]\,\rangle\,\big]{}_{r,s=1}^k\,,
\end{align}
where the diagonal entries of the matrix on the right-hand side are set to~$0$.

One of the most transparent explanations of this Pfaffian structure comes from a remarkable fact that the partition function~$\mathcal{Z}(G)$ of the Ising model can be also written as the Pfaffian of some (real, anti-symmetric) matrix~$\widehat{\mathrm{K}}$, which is a simple transform of the famous \emph{Kac--Ward matrix}~$\mathcal{KW}(G,x):=\mathrm{Id}-\mathrm{T}$, where
\begin{equation}
\label{eq:kac-ward}
\mathrm{T}_{e,e'}=\begin{cases}\exp[\frac{i}{2}(\arg(e')-\arg(e))]\cdot(x_ex_{e'})^{1/2}& \text{if~$e'$~prolongates~$e$,} \\
0 & \text{otherwise,}\end{cases}
\end{equation}
and~$e,e'$ are~\emph{oriented} edges of~$G$. Namely (see~\cite{chelkak-cimasoni-kassel} for more details, including the interpretation via a relevant dimer model on the so-called Fisher graph~$G^\mathrm{F}$, and~\cite{lis-kac-ward} for a streamlined version of classical arguments),
\[
\widehat{\mathrm{K}}=i\mathrm{U}^* \mathrm{J}\cdot \mathcal{KW}(G,x)\cdot \mathrm{U},\quad \text{where}\quad \mathrm{U}=
\mathrm{diag}(\eta_e),\ \ \textstyle \eta_e:=\varsigma\cdot \exp[-\frac{i}{2}\arg(e)],
\]
and~$\mathrm{J}$ is composed of~$2\times 2$ blocks~$\mathrm{J}_{e,e}=\mathrm{J}_{\bar{e},\bar{e}}=0$, $\mathrm{J}_{e,\bar{e}}=\mathrm{J}_{\bar{e},e}=1$, indexed by pairs~$(e,\bar{e})$ of oppositely oriented edges of~$G$. Note the similarity between the definition of~$\eta_e$ and~\eqref{eq:Dirac_spinor}: essentially, one can view the former as an arbitrarily chosen section of the latter, considered on oriented edges of~$G$ instead of~$\Upsilon(G)$.

In other words, the Hamiltonian of the Ising model can be rewritten as a {quadratic} form in \emph{Grassmann variables}~$\phi_e$ (aka free \emph{fermions}), whose (formal) correlators satisfy Pfaffain identities by definition. Then one can check (e.g., see~\cite[Theorem~1.2]{chelkak-cimasoni-kassel}) that these correlators of~$\phi_e$ admit essentially the same combinatorial expansions as the expectations~\eqref{eq:many_chi}. In fact, if all the vertices~$v(c_r)$ and~$v_p$ are pairwise distinct, then all the expectations involved in~\eqref{eq:Pfaff_chi} can be viewed~as (formal) correlators of some Grassmann variables~$\chi_c$ obtained from~$\phi_e$ by a local (namely, block diagonal with blocks indexed by vertices of~$G$) linear change, see~\cite[Section~3.4]{chelkak-cimasoni-kassel}. Therefore, the Pfaffian identities~\eqref{eq:Pfaff_chi} hold true if all~$v(c_r)$ and~$v_p$ are pairwise distinct and this assumption can be removed using the propagation equation~\eqref{eq:propagation_on_corners}, which is satisfied (with respect to each of~$c_r$) by both sides of~\eqref{eq:Pfaff_chi}.

\section{Holomorphic observables in the critical model on isoradial graphs}
\label{sec:hol_obs}

\setcounter{equation}{0}

\subsection{Critical Ising model on isoradial graphs (rhombic lattices)} We now focus on the case when a weighted graph~$G$ is a part of a nice infinite grid, for instance a part of~$\mathbb{Z}^2$. For the homogeneous (i.e., all~$J_e=1$) model on~$\mathbb{Z}^2$, the Kramers--Wannier duality~\eqref{eq:theta_e_def} suggests that the value~$x=\tan\frac{\pi}{8}=\sqrt{2}-1$ corresponds to the critical point of the model and indeed a second order phase transition at~$\beta=-\frac{1}{2}\log(\sqrt{2}-1)$ can be justified in several ways. More generally, one can consider an arbitrary infinite tiling~$\diamondsuit$ of the complex plane by \emph{rhombi} with angles uniformly bounded from below, split the vertices of the bipartite graph~$\Lambda$ formed by the vertices of these rhombi into two \emph{isoradial graphs}~$\Gamma^\bullet$ and~$\Gamma^\circ$, and define the Ising model weights on~$\Gamma^\bullet$ by setting~$x_e:=\tan{\frac{1}{2}\theta_e}$, where~$\theta_e$ is the half-angle of the corresponding rhombus at vertices from~$\Gamma^\bullet$. This model is called a \emph{self-dual Z-invariant Ising model} on isoradial graphs and it can be viewed as a particular point in a family of the so-called Z-invariant Ising models studied by Baxter and parameterized by an elliptic parameter~$k$. Recently, it was shown by Boutillier, de Tili\`ere and Raschel that, similarly to the case of regular lattices, the Z-invariant Ising model on a given isoradial graph exhibits a second order phase transition at its self-dual point~$k=0$, see~\cite{BdTR-ising-16} and references therein for more details.

Below we are mostly interested in the following setup: let~$\Omega\subset\mathbb{C}$ be a bounded simply connected domain and let~$\Omega^\delta=\Omega(G^\delta)$ be a sequence of its polygonal discretizations on isoradial graphs of mesh size~$\delta$ (in other words, the corresponding rhombic lattice is formed by rhombi with side lengths~$\delta$) with~$\delta\to 0$. Below we use the notation~$V^\bullet(\Omega^\delta)$, $V^\circ(\Omega^\delta)$, $V^\Lambda(\Omega^\delta)$, $V^\diamond(\Omega^\delta)$, $V^\curlyvee(\Omega^\delta)$ and~$V^\curlyvee_\varpi(\Omega^\delta)$ for the sets of vertices of~$(G^\delta)^\bullet$, $(G^\delta)^\circ$, $\Lambda(G^\delta)$, $\diamondsuit(G^\delta)$, $\Upsilon(G^\delta)$ and~$\Upsilon_\varpi(G^\delta)$, respectively, and we always assume that the Ising model weighs~$x_e=\tan\frac{1}{2}\theta_e$ are chosen according to the geometry of these isoradial graphs, as explained above. In particular, the reader can think of a sequence of discrete domains~$\Omega^\delta$ drawn on square grids of mesh sizes~$\delta\to 0$ and the homogeneous Ising model with the critical weight~$x=\sqrt{2}\!-\!1$.

\subsection{From Kadanoff--Ceva fermions to discrete holomorphic functions} As already mentioned in Section~\ref{sub:s-hol}, all the Ising model observables of the form
\[
\langle\psi_c\mathcal{O}_\varpi[\mu,\sigma ]\,\rangle := \eta_c\langle\chi_c\mathcal{O}_\varpi[\mu,\sigma ]\,\rangle,\quad \varpi=\{v_1,...,v_{m-1},u_1,...,u_{m-1}\},
\]
considered as functions of~$c\in V^\curlyvee_\varpi(\Omega^\delta)$, always satisfy a three-term propagation equation, which is obtained from~\eqref{eq:propagation_on_corners} by multiplying each of the terms by the Dirac spinor~\eqref{eq:Dirac_spinor}. For the critical Ising model on isoradial graphs, it admits a particularly nice interpretation as the \emph{discrete holomorphicity} of~$F_\varpi(c)$. One has
\[
\textstyle \eta_c=\varsigma\cdot\exp[-\frac{i}{2}\arg(v^\bullet(c)-v^\circ(c))]= \varsigma\delta^{1/2}\cdot(v^\bullet(c)-v^\circ(c))^{-1/2},
\]
and a straightforward computation shows that~\eqref{eq:propagation_on_corners} can be rewritten as
\begin{equation}
\label{eq:psi+psi=psi+psi}
\langle\psi_{c_{00}}\mathcal{O}_\varpi[\mu,\sigma ]\,\rangle+\langle\psi_{c_{11}}\mathcal{O}_\varpi[\mu,\sigma ]\,\rangle
=\langle\psi_{c_{01}}\mathcal{O}_\varpi[\mu,\sigma ]\,\rangle+\langle\psi_{c_{10}}\mathcal{O}_\varpi[\mu,\sigma ]\,\rangle\,.
\end{equation}
Taking the complex conjugate and using~$(v^\bullet(c)\!-\!v^\circ(c))\cdot \eta_c = \varsigma^2\delta\cdot \overline{\eta}_c$, one gets
\[
(v^\circ_0-v^\bullet_0)\Phi(c_{00})+ (v^\bullet_1-v^\circ_0)\Phi(c_{10})+(v^\circ_1-v^\bullet_1)\Phi(c_{11})+ (v^\bullet_0-v^\circ_1)\Phi(c_{01})=0,
\]
where~$\Phi(c):=\langle\psi_c\mathcal{O}_\varpi[\mu,\sigma ]\,\rangle$, which can be thought of as a vanishing discrete contour integral around a given rhombus. However, there exists an even better way to interpret~\eqref{eq:psi+psi=psi+psi}, adopted in~\cite{smirnov-06icm,smirnov-conf-inv-I} and further works. For~$z\in V^\diamond_\varpi(\Omega^\delta)$, denote
\begin{equation}
\label{eq:def_complex_F}
F_\varpi(z)~:=~
\langle(\psi_{c_{00}(z)}\!+\psi_{c_{11}(z)})\mathcal{O}_\varpi[\mu,\sigma]\,\rangle = \langle(\psi_{c_{01}(z)}\!+\psi_{c_{10}(z)})\mathcal{O}_\varpi[\mu,\sigma]\,\rangle.
\end{equation}
It turns out that there exists a natural discrete Cauchy--Riemann operator~$\partial_\Lambda^*$ (acting on functions defined on~$V^\diamond(\Omega^\delta)$ and returning functions on~$V^\Lambda(\Omega^\delta)$) such that
\begin{equation}
\label{eq:dF=0}
\partial_\Lambda^* F_\varpi=0\quad \text{in}\ \ V^\Lambda(\Omega^\delta)\setminus\{v_1,...,v_{m-1},u_1,...,u_{n-1}\}\,.
\end{equation}
More precisely,~$\partial_\Lambda^*$ is the (formally) adjoint operator to
\begin{equation}
\label{eq:def_d_iso}
[\partial_\Lambda H](z):=\frac{1}{2}\biggl[\frac{H(v^\bullet_1(z))\!-\!H(v^\bullet_0(z))}{v^\bullet_1(z)-v^\bullet_0(z)}+ \frac{H(v^\circ_1(z))\!-\!H(v^\circ_0(z))}{v^\circ_1(z)-v^\circ_0(z)}\biggr],
\end{equation}
see~\cite{mercat-CMP,kenyon-inv-02,chelkak-smirnov-adv,chelkak-smirnov} for more details. One of the advantages of~\eqref{eq:dF=0} is that we now have roughly the same number of equations as the number of unknowns~$F(z)$ and do not need to keep track of the fact that the values~$\Phi(c_{pq})$ discussed above have prescribed complex phases (note that the number of unknowns~$\Phi(c)$ is roughly twice the number of vanishing elementary contour integrals around rhombi). However, this information on complex phases is not fully encoded by~\eqref{eq:dF=0}. In fact, a slightly stronger condition holds: for two rhombi~$z,z'$ adjacent to the same~$c\in V^\curlyvee_\varpi(\Omega^\delta)$,
\begin{equation}
\label{eq:s-hol-on-z}
\mathrm{Pr}[F_\varpi(z);\eta_c\mathbb{R}]~=~\mathrm{Pr}[F_\varpi(z');\eta_c\mathbb{R}],
\end{equation}
where~$\mathrm{Pr}[F;\eta\mathbb{R}]:=\overline{\eta}^{\,-1}\mathrm{Re}[\overline{\eta}F]$.
Actually, one can easily see from~\eqref{eq:psi+psi=psi+psi} that both sides of~\eqref{eq:s-hol-on-z} are equal to~$\langle\psi_c\mathcal{O}_\varpi[\mu,\sigma]\,\rangle$ and it is not hard to check that the condition~\eqref{eq:s-hol-on-z} indeed implies~\eqref{eq:dF=0}, see~\cite[Section~3.2]{chelkak-smirnov}.
\begin{remark}
\label{rem:s-hol-fct-spinors}
In~\cite{chelkak-smirnov}, the term \emph{s-holomorpicity} was introduced for \emph{complex-valued functions}~$F(z)$ defined on $V^\diamond(\Omega^\delta)$ and satisfying~\eqref{eq:s-hol-on-z}, in particular to indicate that such functions also satisfy~\eqref{eq:dF=0}. In view of~\eqref{eq:def_complex_F}, this is essentially the same notion as the one introduced in Definition~\ref{def:s-hol} for \emph{real-valued spinors} defined on~$V^\times(\Omega^\delta)$. Still, it is worth mentioning that~\eqref{eq:propagation_on_corners} does not rely upon the very specific (isoradial) choice of the embedding of~$G^\delta$ and the Ising weights, while~\eqref{eq:dF=0} does; see also~\cite[Sections~3.5,3.6]{chelkak-cimasoni-kassel} for a discussion of~\eqref{eq:s-hol-on-z} in the general case.
\end{remark}

\begin{remark}
\label{rem:comb-def-parafermions}
In breakthrough works of Smirnov~\cite{smirnov-06icm,smirnov-conf-inv-I,duminil-smirnov-clay,smirnov-10icm} (see also~\cite{chelkak-smirnov,hongler-smirnov}) on the critical Ising model, discrete holomorphic \emph{fermions} were introduced in a purely combinatorial way, as a particular case of \emph{parafermionic} observables, and the \mbox{s-holo}\-morphicity condition~\eqref{eq:s-hol-on-z} was verified combinatorially as well (e.g., see~\cite[Fig.~5]{smirnov-10icm} or~\cite[Fig.~5]{chelkak-smirnov}). Following this route, more general spinor observables~\eqref{eq:def_complex_F} were also treated combinatorially, and not via the Kadanoff--Ceva formalism, in~\cite{chelkak-izyurov,chelkak-hongler-izyurov}.
\end{remark}

\subsection{Boundary conditions and the key role of the function~$\bm{H_F}$}
\label{sub:boudnary_H}

We now come back to Definition~\ref{def:H_F} suggested in~\cite{smirnov-06icm,smirnov-conf-inv-I} as a crucial tool for the analysis of fermionic observables via boundary value problems for s-holomorphic functions in~$\Omega^\delta$. Let~$F=F^\delta$ be such a function and~$H_F$ be defined via~\eqref{eq:def_of_H} from the corresponding (see Remark~\ref{rem:s-hol-fct-spinors}) real-valued s-holomorphic spinor on~$V^\times(\Omega^\delta)$. A straightforward computation shows that, for each~$z\in V^\diamond(\Omega^\delta)$, one has
\begin{equation}
\label{eq:H-def-as-Im}
\begin{array}{l}
H_F(v^\bullet_1)\!-\!H_F(v^\bullet_0)=\mathrm{Re}[\,\overline{\varsigma}^2\delta^{-1}(F(z))^2(v^\bullet_1\!-\!v^\bullet_0)]= \mathrm{Im}[\delta^{-1}(F(z))^2(v^\bullet_1\!-\!v^\bullet_0)],\\
H_F(v^\circ_1)\!-\!H_F(v^\circ_0)=\mathrm{Re}[\,\overline{\varsigma}^2\delta^{-1}(F(z))^2(v^\circ_1\!-\!v^\circ_0)]= \mathrm{Im}[\delta^{-1}(F(z))^2(v^\circ_1\!-\!v^\circ_0)],
\end{array}
\end{equation}
provided that the global prefactor in the definition~\eqref{eq:Dirac_spinor} is chosen\begin{footnote}{This is a matter of convenience. E.g., the notation in~\cite{hongler-smirnov,hongler-kytola,chelkak-hongler-izyurov,gheissari-hongler-park} corresponds to~$\varsigma=i$.}\end{footnote} as~$\varsigma=e^{i\frac{\pi}{4}}$. In other words, the real-valued function~$H_F$ turns out to be a proper discrete analogue of the expression~$\int\mathrm{Im}[(F(z))^2dz]$. Even more importantly, $H_F$ encodes the boundary conditions of the Ising model in the form of Dirichlet boundary conditions. Namely, one can choose an additive constant in its definition so that
\begin{equation}
\label{eq:H_bdry_cond}
\begin{array}{l}
\bullet~\text{$H_F=0$ and~$\partial_{\nu_{\mathrm{in}}}H\ge 0$ on all the wired boundary arcs,}\\
\bullet~\text{$H_F$ is constant and~$\partial_{\nu_\mathrm{in}}H\le 0$ on each of the free boundary arcs,}
\end{array}
\end{equation}
where~$\partial_{\nu_\mathrm{in}}$ stands for a discrete analogue of the inner normal derivative, see~\cite{chelkak-smirnov,izyurov-free}. In particular, the values of the s-holomorphic spinor~$F$ at two endpoints of a free arc have the same absolute value (which, according to~\eqref{eq:def_of_H}, is equal to the square root of the value of~$H$ on this free arc) and in fact one can also match their signs by tracking the Dirac spinor~\eqref{eq:Dirac_spinor} along the boundary; see~\cite{izyurov-free,chelkak-hongler-izyurov-mixed} for more details.

One can also check that, for~$\varpi=\{v_1,...,v_{m-1},u_1,...,u_{n-1}\}$ and~$H_{F_\varpi}$ obtained from the s-holomorphic function~$F_\varpi$ given by~\eqref{eq:def_complex_F} (or, equivalently, from the \mbox{s-ho}lo\-morphic spinor~$\langle\chi_c\mathcal{O}_\varpi[\mu,\sigma]\rangle$), the following is fulfilled:
\begin{itemize}
\item the minimum of~$H_F$ in a vicinity of each of~$v_p$ is attained at the boundary,
\item the maximum of~$H_F$ in a vicinity of each of~$u_q$ is attained at the boundary.
\end{itemize}

Imagine now that instead of a sequence of discrete s-holomorphic functions~$F^\delta$ on~$\Omega^\delta$ we had a continuous holomorphic function~$f:\Omega\to\mathbb{C}$ such that the harmonic function~$h_f:=\int \mathrm{Im}[(f(z))^2dz]$ satisfied the conditions listed above. We could then hope to identify~$f$ as a solution to such a boundary value problem in~$\Omega$, provided that this solution is unique (up to a normalization to be fixed). This not necessarily true in full generality (e.g., some degeneracy might appear in presence of several spins and several disorders in~$\mathcal{O}_\varpi[\mu,\sigma]$) but there are enough situations in which a relevant uniqueness theorem `in continuum' is fulfilled and thus one can hope to prove the convergence of~$F^\delta$ to~$f$ as~$\delta\to 0$; see Section~\ref{sub:conv_s-hol} and~\cite{chelkak-hongler-izyurov-mixed} for more details.

\subsection{Smirnov's sub-/super-harmonicity} The interpretation of \mbox{s-holo}\-morphic functions~$F^\delta$ as solutions to discrete boundary value problems described above is implicitly based on the idea that one can think of~$H_{F^\delta}$ as of a discrete harmonic function on~$\Lambda(\Omega^\delta)$. However, this cannot be true literally: the functions~$(F^\delta)^2$ are not discrete holomorphic, thus there is no hope that~$H_{F^\delta}$ are exactly discrete harmonic. Nevertheless, there is a miraculous positivity phenomenon, first observed in~\cite{smirnov-06icm,smirnov-conf-inv-I} on the square grid and later generalized to the isoradial setup in~\cite{chelkak-smirnov}.
\begin{lemma}[see {\cite[Lemma~3.8]{smirnov-conf-inv-I}}, {\cite[Proposition~3.6(iii)]{chelkak-smirnov}}]
\label{lemma:positivity}
Let~$F$ be defined on quads surrounding a given vertex~$v\in V^\Lambda(\Omega^\delta)$ and satisfy the \mbox{s-holo}\-morphicity identities~\eqref{eq:s-hol-on-z}. If~$H_F$ is constructed via~\eqref{eq:H-def-as-Im} (or, equivalently, via~\eqref{eq:def_of_H}), then
\begin{equation}
\label{eq:positivity_iso}
[\Delta^\bullet H_F](v)\ge 0\ \ \text{if}\ \ v\in V^\bullet(\Omega^\delta),\quad\text{and}\quad [\Delta^\circ H_F](v)\le 0\ \ \text{if}\ \ v\in V^\circ(\Omega^\delta),
\end{equation}
where the Laplacian operator~$\Delta^\bullet$ (and, similarly,~$\Delta^\circ$) is defined as
\begin{equation}
\label{eq:def_Delta_iso}
\textstyle [\Delta^\bullet H](v)~:=~\sum_{v_1:v_1\sim v} \tan\theta_{(vv_1)}\cdot(H(v_1)\!-\!H(v)).
\end{equation}
\end{lemma}

It is easy to check that both discrete operators~$\Delta^\bullet$ and~$\Delta^\circ$ approximate the standard Laplacian as~$\delta\to 0$ in a quite strong (local) sense; see~\cite{chelkak-smirnov-adv} for details. Recall that, according to Definition~\ref{def:H_F} the discrete \emph{subharmonic} function~$H_{F^\delta}|_{V^\bullet(\Omega^\delta)}$ is pointwise (i.e., at any pair of neighboring vertices~$v^\bullet(c)$ and~$v^\circ(c)$) greater or equal to the discrete \emph{superharmonic} function~$H_{F^\delta}|_{V^\circ(\Omega^\delta)}$. Therefore, provided that the difference~$\delta H_{F^\delta}(v^\bullet(c)-\delta H_{F^\delta}(v^\circ(c))=(F^\delta(c))^2$ is small inside of~$\Omega^\delta$, both~$\delta H_{F^\delta}|_{V^\bullet(\Omega^\delta)}$ and~$\delta H_{F^\delta}|_{V^\circ(\Omega^\delta)}$ should have the same \emph{harmonic} limit as~$\delta\to 0$.

\begin{remark} Still, there is a question of how to show that~$F^\delta$ is small. In the pioneering work~\cite{smirnov-06icm} devoted to basic fermionic observables on~$\mathbb{Z}^2$, this fact was derived from monotonicity arguments and the magnetization estimate; see~\cite[Lemma~A.1]{smirnov-conf-inv-I}. Shortly afterwards, an \emph{a priori} regularity theory for functions~$H_F$ were developed in~\cite[Section~3]{chelkak-smirnov}. These a priori estimates were later applied to various Ising-model observables, in particular when no simple monotonicity arguments are available.
\end{remark}

\begin{remark}
\label{rem:no-go-positivity}
It is commonly believed that one cannot \emph{directly} generalize Lemma~\ref{lemma:positivity} for more general graphs or Ising-model weights~$x_e=\frac{1}{2}\tan\theta_e$: the existence of \emph{some} Laplacians~$\Delta^\bullet$ and~$\Delta^\circ$ such that~\eqref{eq:positivity_iso} holds true seems to be equivalent to the conditions~$\sum_{v_1:v_1\sim v}\arctan x_{(vv_1)}\in \frac{\pi}{2}+\pi\mathbb{Z}$, which lead to an isoradial embedding of~$(G^\bullet,G^\circ)$, possibly with~$(2\pi+4\pi\mathbb{Z})$ conical singularities at vertices.
\end{remark}

\section{Convergence of correlations}
\label{sec:conv_corr}
\setcounter{equation}{0}

In this section we very briefly discuss the convergence results for correlation functions (fermions, energy densities, spins etc) obtained during the last decade for the critical model on~$\mathbb{Z}^2$ (see also Remark~\ref{rem:limits-on-iso} on the critical Z-invariant model). We refer the interested reader to upcoming~\cite{chelkak-hongler-izyurov-mixed}, see also~\cite[Section~4]{chelkak-16ecm}.

\subsection{Convergence of s-holomorphic observables} \label{sub:conv_s-hol} As already mentioned above, techniques developed in~\cite{chelkak-smirnov} essentially allow one to think of sub-/super-harmonic functions $H^\delta:=\delta^{-1}H_{F^\delta}$ as of harmonic ones. In particular, a uniform boundedness of the family~$H^\delta$ on an open set implies that both~$H^\delta$ and the original observables~$F^\delta$ are equicontinuous on compact subsets of this open set.

Imagine now that we want to prove the convergence of (properly normalized) observables~$F^\delta_\varpi$,~$\varpi=\{v_1,...,v_{m-1},u_1,...,u_{n-1}\}$ as~$\delta\to 0$. Assuming that the corresponding functions~$H^\delta_\varpi$ are uniformly bounded away from~$v_p$'s and~$u_q$'s, the Arzel\`a--Ascolli theorem ensures that, \emph{inside of}~$\Omega\setminus\varpi$, at least subsequential limits~$H^\delta_\varpi\to h_\varpi$ and $F^\delta_\varpi\to f_\varpi$ exist. Since the functions~$F^\delta_\varpi$ are discrete holomorphic, their limit~$f_\varpi$ is a holomorphic function and one also has~$h_\varpi=\int\mathrm{Im}[(f_\varpi(z))^2dz]$. Moreover, one can show that the \emph{boundary conditions}~\eqref{eq:H_bdry_cond} survive as~$\delta\to 0$, see~\cite[Remark~6.3]{chelkak-smirnov} and~\cite{izyurov-free}. Clearly,~$h_\varpi$ also inherits from~$H^\delta_\varpi$ the semi-boundedness from below near~$v_p$ and from above near~$u_q$. Thus, only two questions remain:

(i)\phantom{i} to show that~$f$ and~$h$ are uniquely characterized by the above properties;

(ii) to justify that the functions~$\delta^{-1}H_\delta$ are uniformly bounded away from~$\varpi$.

\noindent In general (i.e., in presence of several disorders and spins in~$\mathcal{O}_\varpi[\mu,\sigma]$), the uniqueness (i) may fail. Fortunately, there exists a principal setup in which it holds true:
\begin{equation}
\label{eq:basic-F}
F^\delta_\varpi(z):=\frac{\delta^{-1}\langle\psi(z)\,\chi_{d}\,\sigma_{u_1}...\sigma_{u_{n-2}}\rangle} {\langle\sigma_{u_1}...\sigma_{u_{n-2}}\rangle}\,,\quad \psi(z):=\psi_{c_{00}(z)}\!+\!\psi_{c_{11}(z)}\,,
\end{equation}
where~$d$ and~$u_1,...,u_{n-2}$ are assumed to approximate distinct points of~$\Omega$ as~$\delta\to 0$.
(Recall also that the roles of spins and disorders are not fully symmetric since our standard boundary conditions are not fully symmetric under the Kramers--Wannier duality). Note that such functions~$F^\delta_\varpi$ have `standard' discrete singularities at~$d$, leading to a `standard' singularity (simple pole with a fixed residue) of~$f_\varpi$ at~$d$.

Finally, a useful trick allowing to deduce (ii) from (i) was suggested in~\cite[Section~3.4]{chelkak-hongler-izyurov}: if the functions~$H^\delta_\varpi$ were unbounded, then one could renormalize them in order to obtain non-trivial limits~$\widetilde{f}_\varpi$ and $\widetilde{h}_\varpi=\int\mathrm{Im}[(\widetilde{f}_\varpi(z))^2dz]$, which would solve the same boundary value problem as~$f_\varpi$ and~$h_\varpi$ but without the pole at~$d$, killed by such an additional renormalization. Due to (i), this boundary value problem has no nontrivial solution, which gives a contradiction; see also~\cite{chelkak-hongler-izyurov-mixed}.

\subsection{Fusion: from s-holomorphic observables to~$\bm{\varepsilon}$,~$\bm{\sigma}$ and~$\bm{\mu}$}
Once the convergence of observables~\eqref{eq:basic-F} is established, one can first use the Pfaffian structure~\eqref{eq:Pfaff_chi} of fermionic correlators to obtain the convergence of the quantities
\begin{equation}
\label{eq:chi-chi-sigma-sigma}
\frac{\langle\chi_{c_1}...\chi_{c_k}\sigma_{u_1}...\sigma_{u_{n-k}}\rangle}{\langle\sigma_{u_1}...\sigma_{u_{n-k}}\rangle}~=~ \mathrm{Pf}\biggl[\frac{\langle\chi_{c_r}\chi_{c_s}\sigma_{u_1}...\sigma_{u_{n-k}}\rangle} {\langle\sigma_{u_1}...\sigma_{u_{n-k}}\rangle}\biggr]_{s,r=1}^k\,.
\end{equation}
In particular, the simplest case~$n-k=0$ allows one to study~\cite{hongler-smirnov,hongler-thesis} the scaling limit of correlations of the \emph{energy density} field, defined for~$z\in V^\diamond(\Omega^\delta)$ as
\[
\varepsilon_z~:=~\sigma_{v^\circ_0(z)}\sigma_{v^\circ_1(z)}\!-\!2^{-\frac{1}{2}}\,=\,\pm\,2^{-\frac{1}{2}} \chi_{c_{00}(z)}\chi_{c_{11}(z)}\,=\, 2^{-\frac{1}{2}}\!-\!\mu_{v^\bullet_0(z)}\mu_{v^\bullet_1(z)}.
\]
Note that a careful analysis of the singularity of~\eqref{eq:basic-F} as~$z\to d$ is required in order to handle the next-to-leading term~$\varepsilon_z$ appearing in the fusion of two fermions~$\chi_c\chi_d$.

In the same spirit, recursively analyzing the behavior of~\eqref{eq:chi-chi-sigma-sigma} near discrete singularities as~$c_{k-s}\!\to u_{n-k-s}$, $s=0,...,r\!-\!1$, one obtains scaling limits of the ratios
\begin{equation}
\label{eq:chi-mu-sigma}
\langle\chi_{c_1}...\chi_{c_{k-r}}\mu_{u^\bullet{}_{\!\!\!n-k-r+1}}...\mu_{u^\bullet{}_{\!\!\!n-k}}\sigma_{u_1}...\sigma_{u_{n-k-r}}\rangle \cdot \langle\sigma_{u_1}...\sigma_{u_{n-k}}\rangle^{-1},
\end{equation}
where~$u^\bullet_p\in V^\bullet(\Omega^\delta)$ denotes one of the neighboring to~$u_p$ vertices of~$\Gamma^\bullet$. The remaining ingredient is the convergence of the denominators~$\langle\sigma_{u_1}...\sigma_{u_{n-k}}\rangle$, which do not contain any discrete holomorphic variable. Such correlations were treated in~\cite{chelkak-hongler-izyurov} using the following idea: fusing~$\chi_d$ and~$\sigma_{u_1}$ in~\eqref{eq:basic-F} one gets the observable
\[
\langle\psi(z)\mu_{u^\bullet_{1}}\sigma_{u_2}...\sigma_{u_{n-2}}\rangle \cdot \langle\sigma_{u_1}...\sigma_{u_{n-2}}\rangle^{-1}
\]
and then analyzes its behavior (as that of a discrete holomorphic function in~$z$) near the vertex~$u_1^\bullet\sim u_1$. This analysis provides an access to the ratio
\begin{equation}
\label{eq:d-multi-spin}
\langle\sigma_{u^\circ_1}\sigma_{u_2}...\sigma_{u_{n-2}}\rangle\cdot \langle\sigma_{u_1}...\sigma_{u_{n-2}}\rangle^{-1},\quad
\end{equation}
where~$u^\circ_1\sim u^\bullet_1\sim u^{\phantom{\circ}}_1$ can be any neighboring to~$u_1$ face of $\Omega^\delta$. The next-to-leading term in the asymptotics of this ratio as~$\delta\to 0$ encodes the discrete spatial derivative of~$\log\langle\sigma_{u_1}...\sigma_{u_{n-2}}\rangle$, which eventually allows one to identify its scaling limit by fusing spins to each other and thus reducing~$n$; see~\cite{chelkak-hongler-izyurov} and~\cite{chelkak-hongler-izyurov-mixed} for more details.

\begin{remark} The above scheme allows one to prove the convergence of arbitrary correlations of disorders, spins and fermions under \emph{standard} boundary conditions in~$\Omega^\delta$. In fact, it can be further generalized to include (a) the (common) spin~$\sigma_{u_\mathrm{out}}$ of the wired boundary arcs; (b) one or several disorders assigned to free boundary arcs; (c) fermions on the boundary of~$\Omega^\delta$. In particular, this allows one to handle an arbitrary mixture of `$+$', `$-$' and `free' boundary conditions; see~\cite{chelkak-hongler-izyurov-mixed} for details.
\end{remark}

\begin{remark}
\label{rem:limits-on-iso} The convergence results for s-holomorphic observables~\eqref{eq:basic-F} and the energy-density correlations can be also proved in the isoradial setup ad verbum. Nevertheless, the passage to~\eqref{eq:chi-mu-sigma} and, especially, the analysis of the spatial derivatives~\eqref{eq:d-multi-spin} of spin correlations require some additional work. We believe that all the key ingredients can be extracted from~\cite{dubedat-cauchy-riemann} but it is worth mentioning that such a generalization has not appeared yet even for the honeycomb/triangular grids.
\end{remark}

\subsection{More CFT on the lattice}
\label{sub:lattice-CFT}
From the Conformal Field Theory perspective (e.g., see~\cite{mussardo-book}), the scaling limits of correlations of fermions, spins, disorders and energy-densities are those of \emph{primary fields} (non-local ones in the case of~$\psi$ and~$\mu$). It is a subject of the ongoing research to construct the corresponding CFT, as fully as possible, directly on the lattice level (and not in the limit as~$\delta\to0$). Below we mention several results and research projects in this direction:

\smallskip

-- The analysis of general \emph{lattice fields} (i.e., functions of several neighboring spins) was started in~\cite{gheissari-hongler-park}, though at the moment only their leading terms, which converge to either~$1,\varepsilon,\sigma$ or the spatial derivative of~$\sigma$ in the scaling limit, are treated.

-- The action of the \emph{Virasoro algebra} on such lattice fields was recently defined in~\cite{hongler-kytola-viklund}, via the so-called Sugawara construction applied to discrete fermions.

-- An `infinitesimal non-planar deformation' approach to the \emph{stress-energy tensor} of the Ising model on faces of the honeycomb grid was recently suggested in~\cite{chelkak-glazman-smirnov}.

\smallskip

We also believe that one can use s-embeddings (see Section~\ref{sec:S-embeddings}) of weighted planar graphs~$(G,x)$ to properly interpret an infinitesimal change of the Ising weights~$x_e$ as a \emph{vector-field} in~$\mathbb{C}$, thus providing yet another approach to the stress-energy tensor.

\section{Interfaces and loop ensembles}
\setcounter{equation}{0}

In the 20th century, the (conjectural) conformal invariance of critical lattice models was typically understood via the scaling limits of correlation functions, which are fundamental objects studied by the classical CFT~\cite{mussardo-book}, though see~\cite{aizenman-burchard} and~\cite{langlands-et-al-ising}.
The introduction of SLEs by Schramm~\cite{schramm-00}, and later developments on CLEs (Conformal Loop Ensembles, see~\cite{sheffield,sheffield-werner-cle, miller-sheffield-werner-cle-percolation} and references therein) provided quite a different perspective of studying the \emph{geometry} of either particular interfaces (e.g., domain walls in the critical Ising model) -- conjecturally converging to SLEs -- or the full collection of interfaces -- conjecturally converging to CLEs -- as~$\Omega^\delta\!\to\Omega$.

For the critical Ising model, both in its classical and the so-called random-cluster (or Fortuin--Kastelen, see below) representations, the convergence of interfaces generated by Dobrushin boundary conditions was obtained already in the pioneering work of Smirnov (see~\cite[Theorem~1]{smirnov-06icm} and references therein) via the application of the so-called \emph{martingale principle} (e.g., see~\cite[Section~5.2]{smirnov-06icm}) and the convergence results for basic fermionic observables~\cite{smirnov-conf-inv-I,chelkak-smirnov}. However, the improvement of the \emph{topology of convergence} from that of functional parameters (aka driving forces) in the Loewner equation to the convergence of curves themselves required some additional efforts: the appropriate framework was provided by Kemppainen and Smirnov in~\cite{kemppainen-smirnov}; see also~\cite{5-authors} and references therein. Basing on this framework, the convergence of the full \emph{branching tree} of interfaces in the FK-representation (announced in~\cite{smirnov-06icm,smirnov-conf-inv-I}) to CLE(16/3) was eventually justified by the same authors in~\cite{kemppainen-smirnov-ii,kemppainen-smirnov-iii}. In parallel, an exploration algorithm aiming at the proof of the convergence of the classical domain walls ensemble to CLE(3) was suggested in~\cite{hongler-kytola} and later, via the convergence of the so-called \emph{free arc ensemble} established in~\cite{benoist-duminil-hongler}, this convergence to CLE(3) has also been justified by Benoist and Hongler~\cite{benoist-hongler}. Recently, another approach to derive the convergence of the domain walls loop ensemble to CLE(3) from that of the random cluster one to CLE(16/3) was suggested in~\cite{miller-sheffield-werner-cle-percolation}.

Certainly, it is absolutely impossible to provide details of these advanced developments in a short note, thus we refer the interested reader to the original articles mentioned above and hope that such a survey will appear one day. Below we only emphasize several ingredients coming from the complex analysis side and indicate the role played by the Ising model observables discussed above.

\subsection{FK-Ising (random cluster) representation and crossing probabilities}
Recall that the random-cluster representation of the critical Ising model on~$\mathbb{Z}^2$ (e.g., see~\cite[Section~2.3]{smirnov-06icm}) is a probability measure on the configurations of edges of~$\Gamma^\circ$ (each edge is declared open or closed), proportional to
\begin{equation}
\label{eq:FK-measure}
x_{\mathrm{crit}}^{\#\mathrm{closed~edges}}\,(1\!-\!x_{\mathrm{crit}}^{\vphantom{\#}})^{\#\mathrm{open~edges}}\cdot 2^{\#\mathrm{clusters}}\ \sim \ \sqrt{2}^{\#\mathrm{loops}},
\end{equation}
where~$\#\mathrm{clusters}$ stands for the number of connected components in a given configuration and $\#\mathrm{loops}$ denotes the number of loops separating these clusters of vertices of~$\Gamma^\circ$ from dual ones, living on~$\Gamma^\bullet$ (and formed by the edges of~$\Gamma^\bullet$ dual to the closed ones of~$\Gamma^\circ$); e.g., see~\cite[Fig.~3]{smirnov-06icm}. Through the \emph{Edward--Sokal coupling}, it is intimately related to the original spin model: to obtain a random cluster configuration from~$\sigma$, one tosses a (biased according to the Ising weight~$x_e$) coin for each edge~$e$ of~$\Gamma^\circ$ connecting aligned spins; and, inversely, tosses a fair coin to assign a~$\pm 1$ spin to each of the clusters of a given random cluster configuration. In particular,
\begin{equation}
\label{eq:spin=connectivity}
\mathbb{E}^{\circ}[\sigma_{u_1}...\sigma_{u_n}]=\mathbb{P}^{\mathrm{FK}}[\text{each~cluster~contains~an~even~number~of~$u_1,...,u_n$}].
\end{equation}
\begin{remark}
\label{rem:FK-vs-loops}
The fact that the probability measure~\eqref{eq:FK-measure} is proportional to~$\sqrt{2}^{\#\mathrm{loops}}$ relies upon the duality of the graphs~$\Gamma^\bullet$ and~$\Gamma^\circ$, as embedded into a~\emph{sphere}.
 Recall that we consider {all} the vertices of~$\Gamma^\circ$ on wired arcs as a \emph{single} `macro-vertex' while in~$\Gamma^\bullet$ there is one `macro-vertex' for \emph{each} of the free arcs, so this duality holds.
\end{remark}
A particularly important application of the identity~\eqref{eq:spin=connectivity} (applied to disorders on~$\Gamma^\bullet$ rather than to spins on~$\Gamma^\circ$) appears in the \emph{quadrilateral} setup, when there are two wired $(ab),(cd)$ and two free~$(bc),(da)$ arcs on~$\partial\Omega^\delta$. Namely, one has
\begin{equation}
\label{eq:crossing-proba}
\langle \mu_{(da)}\mu_{(bc)}\rangle ~ = ~
\mathbb{P}^{\mathrm{FK}}[(da)\leftrightarrow (bc)] ~=~ \varrho\,\big(\mathbb{P}^{\mathrm{loops}}[(da)\leftrightarrow (bc)]\big)\,,
\end{equation}
where~$\varrho(p):=p\cdot (p+\sqrt{2}(1\!-\!p))^{-1}$ and~$\mathbb{P}^{\mathrm{loops}}\sim \sqrt{2}^{\#\mathrm{loops}}$ denotes the probability measure on FK-Ising configurations, where only the loops lying inside of~$\Omega^\delta$ are counted. In other words, in this \emph{self-dual} setup the two wired arcs are not connected, similarly to the free ones; cf. Remark~\ref{rem:FK-vs-loops}. Note that, though~$\mathbb{P}^{\mathrm{loops}}$ does not literally correspond to any FK-Ising measure, it is absolutely continuous with respect to~$\mathbb{P}^{\mathrm{FK}}$, with the density proportional to~$1+(\sqrt{2}-1)\cdot \bm{1}[(da)\leftrightarrow (bc)]$.

\subsection{Crossing estimates and tightness}
When studying the convergence of random curves~$\gamma^\delta$ as~$\delta\to 0$ (e.g.,~$\gamma^\delta$ can be an interface generated by Dobrushin boundary conditions~\cite{5-authors}, or a branch of the FK-Ising tree~\cite{kemppainen-smirnov-ii,kemppainen-smirnov-iii} or of the free arc ensemble~\cite{benoist-duminil-hongler,benoist-hongler}), an important step is to establish the tightness of this family in the topology induced by the natural metric on the space of curves considered up to reparametrizations.
Departing from the classical results of~\cite{aizenman-burchard}, which appeared even before the introduction of SLE, Kemppainen and Smirnov~\cite{kemppainen-smirnov} showed that a very weak estimate on the probability of an \emph{annulus crossing} implies not only the tightness of~$\gamma^\delta$ themselves but also the tightness of the corresponding random driving forces in the Loewner equations and a uniform bound on exponential moments of these driving forces. Below we only emphasize the following two points of the study of crossing estimates; see the original article~\cite{kemppainen-smirnov} for more details.

\smallskip

$\bullet$ The annulus crossing estimate mentioned above (\emph{geometric} Condition~G in~\cite{kemppainen-smirnov}) was shown to be equivalent to a similar estimate on crossings of general topological quadrilaterals (\emph{conformal} Condition~C in~\cite{kemppainen-smirnov}), uniform in the \emph{extremal length} of a quadrilateral. The latter is conformally invariant by definition, which makes the framework developed in~\cite{kemppainen-smirnov} extremely well-suited to the SLE theory.

$\bullet$ Provided suitable monotonicity (with respect to the position of the boundary and the boundary conditions) arguments are available, it is enough to obtain the required crossing estimates for \emph{two standard quadrilaterals} only, with alternating (i.e., wired/free/wired/free or `$+/-/+/-$', respectively) boundary conditions.

\smallskip

For the FK-Ising model, it is then enough to prove a uniform (in~$\delta$) lower bound for the quantities~\eqref{eq:crossing-proba}. In fact, using techniques described in Section~\ref{sec:conv_corr}, one can even find the limit of these correlations as~$\delta\to 0$ (see also~\cite[Theorem~6.1]{chelkak-smirnov} for a shortcut suggested earlier by Smirnov, which reduces~\eqref{eq:crossing-proba} to the analysis of basic s-holomorphic observables). A similar crossing estimate for the spin-Ising model can be then easily deduced via the Edwards--Sokal coupling and another application of the FKG inequality, see \cite[Remark~4]{5-authors} and~\cite[Section~5.3]{chelkak-duminil-hongler} (a more involved but self-contained argument can be found in~\cite[Section~4.2]{kemppainen-smirnov}).

\begin{remark}
\label{rem:strong-RSW}
In absence of the required monotonicity arguments, one can always use the `strong' RSW-type theory developed for the critical FK-Ising model in~\cite{chelkak-duminil-hongler} (basing, in particular, on techniques from~\cite{chelkak-toolbox}) to verify Condition~C of~\cite{kemppainen-smirnov}.
\end{remark}

\subsection{Martingale observables and convergence of interfaces}
According to the classical martingale principle, aiming to describe the scaling limit of a single interface~$\gamma^\delta$ running in~$\Omega^\delta$ from a marked boundary point~$\gamma^\delta(0)=a^\delta$, one can (try to) find an observable~$M^\delta(z)=M_{\Omega^\delta,a^\delta}(z)$ such that, for each~$z\in\Omega^\delta$, the value~$M_{\Omega^\delta\setminus \gamma^\delta[0;t],\gamma^\delta(t)}(z)$ is a martingale with respect to the filtration generated by~$\gamma^\delta[0;t]$. Provided that the scaling limit~$M(z)$ of~$M^\delta(z)$ as~$\delta\to 0$ can be identified and the tightness conditions discussed in the previous section are fulfilled, one thus gets a \emph{family} (indexed by~$z\in\Omega$) of martingales~$M_{\Omega\setminus\gamma[0;t],\gamma(t)}(z)$ for a (subsequential) limit~$\gamma$ of~$\gamma^\delta$, which is enough to identify its law as that of SLE($\kappa$). Note, however, that some additional analysis is usually required when working with more general SLE($\kappa$,$\rho$) curves in order to control their behavior on the set of times when the corresponding Bessel process hits~$0$, see~\cite{kemppainen-smirnov-ii,kemppainen-smirnov-iii} and~\cite{benoist-duminil-hongler,benoist-hongler} for details.

The convergence results on correlation functions described above provide such martingale observables~$M^\delta(z)$, amenable to the analysis in the limit~$\delta\to 0$, for a huge variety of setups, in both spin- and FK-representations of the model. E.g.,

\smallskip

$\bullet$\ \ $\langle\psi(z)\chi_a\rangle/\langle\mu_b\mu_a\rangle$ is a martingale for the \emph{domain wall} (converging~\cite{5-authors}, as $\delta\!\to\!0$, to SLE(3)) generated by Dobrishin (`$-$' on~$(ab)$,~`$+$' on~$(ba)$) boundary conditions;

$\bullet$\ \ more generally, $\langle\psi(z)\chi_a\rangle/\langle\mu_{(bc)}\mu_a\rangle$ is a martingale for the domain wall emanating from~$a$ under~`$+/-/\mathrm{free}$' boundary conditions (this interface converges~\cite{hongler-kytola,izyurov-free} to the dipolar SLE(3,-3/2,-3/2), an important building block of~\cite{benoist-duminil-hongler,benoist-hongler});

$\bullet$\ \ $\langle\psi(z)\mu_{(ba)}\sigma_{(ab)}\rangle$ is a martingale for the \emph{FK-interface} generated by Dobrushin (wired on~$(ab)$, free on~$(ba)$) boundary conditions (this is the observable introduced by Smirnov~\cite{smirnov-06icm,smirnov-conf-inv-I} to prove the convergence of this interface to SLE(16/3)).

\begin{remark} Following~\cite{smirnov-06icm,smirnov-conf-inv-I,chelkak-smirnov,duminil-smirnov-clay}, these martingale observables are usually defined in a purely combinatorial way in the existing literature (and not via the Kadanoff--Ceva formalism, cf. Remark~\ref{rem:comb-def-parafermions}). One of the advantages of this approach is that the martingale property becomes a triviality. On the other hand, the origin of the crucial s-holomorphicity property becomes less transparent.
\end{remark}

\subsection{Convergence of loop ensembles}
As already mentioned above, we do not intend to overview the details of~\cite{kemppainen-smirnov-ii,kemppainen-smirnov-iii} (convergence of the FK-Ising loop ensemble to CLE(16/3)) and of~\cite{benoist-duminil-hongler,benoist-hongler} (convergence of the domain walls ensemble to CLE(3))  in this essay and refer the interested reader to the original articles. In both projects, some iterative procedure is used: a branching exploration of loops in the former (which is a prototype of the branching SLE-tree from~\cite{sheffield}) and an alternate exploration of FK-Ising clusters and free arc ensembles (proposed in~\cite{hongler-kytola}) in the latter. Because of the hierarchial nature of these algorithms, the following subtlety arises: even if the convergence of each of FK-interfaces to SLE(16/3) curves is known, one still needs to control the behavior of its double points, which split the current domain into smaller pieces to be explored (and also of the endpoints of arcs constituting the free arc ensemble in~\cite{benoist-duminil-hongler}), in order to guarantee that the exploration algorithm in discrete does not deviate too much from the continuum one. This is done\begin{footnote}{The results of~\cite{kemppainen-smirnov-ii,kemppainen-smirnov-iii} can be made self-contained though the current version relies upon~\cite{chelkak-duminil-hongler} (while checking the required crossing estimates for the full tree of interfaces in the proof of \cite[Theorem~3.4]{kemppainen-smirnov-ii}) in order to lighten the presentation (S.~Smirnov, private communication).}\end{footnote}
using the strong RSW theory~\cite{chelkak-duminil-hongler}, which guarantees that all such `pivotal' points in continuum are the limits of those in discrete, uniformly in the shape of subdomains obtained along the way and boundary conditions, cf. Remark~\ref{rem:strong-RSW}.

\section{Towards universality beyond isoradial graphs}
\label{sec:S-embeddings}
\setcounter{equation}{0}

As discussed above, the fundamental questions of convergence and conformal invariance of the critical Ising model on the square grid are now relatively well understood, both for correlation functions and loop ensembles. Moreover, a great part, if not all, of these results can be generalized from~$\mathbb{Z}^2$ to the critical \emph{Z-invariant} model using the already existing techniques. Nevertheless, this does not give a fully satisfactory understanding of the \emph{universality} phenomenon since the cornerstone `sub-/super-harmonicity' Lemma~\ref{lemma:positivity} does not admit a direct generalization beyond the isoradial case, see Remark~\ref{rem:no-go-positivity}. E.g., even the convergence of Ising interfaces on {doubly-periodic} weighted graphs has never been treated though the criticality condition on the weights in this case is well known; see~\cite{cimasoni-duminil} and references therein.

The main purpose of this section is to discuss a new class of embeddings of \emph{generic} planar weighted graphs carrying the  Ising model into the complex plane, with the emphasis on an analogue of the `s-Lemma'~\ref{lemma:positivity}. Below we only sketch some important features of the construction, details will appear elsewhere. Independently of our paper, a special class of s-embeddings -- circle patterns -- was studied by Lis in~\cite{lis-circle-patterns} and the \emph{criticality} was proven in the case of uniformly bounded faces.

\subsection{S-embeddings of weighted planar graphs} Below we adopt the notation of Section~\ref{sec:formalism}. To construct a particular embedding into~$\mathbb{C}$ of a given planar weighted graph~$(G,x)$ we choose \emph{a pair~$\mathcal{F}_1,\mathcal{F}_2$ of s-holomorphic spinors} on~$\Upsilon^\times(G)$ and denote~$\mathcal{F}:=\mathcal{F}_1+i\mathcal{F}_2$. For instance, one can imagine one of the following setups:
\begin{itemize}
\item $G$ is an infinite graph, $\mathcal{F}_1,\mathcal{F}_2$ are s-holomorphic everywhere on~$\Upsilon^\times(G)$;
\item $G$ has the topology of a sphere, the condition~\eqref{eq:propagation_on_corners} is relaxed on a `root' quad (note that $\mathcal{F}_1,\mathcal{F}_2$ cannot be s-holomorphic everywhere on~$\Upsilon^\times(G)$);
\item a finite graph~$G$, the s-holomorphicity of~$\mathcal{F}_1,\mathcal{F}_2$ is relaxed at the boundary.
\end{itemize}
Clearly, the propagation equation~\eqref{eq:propagation_on_corners} still holds for the \emph{complex-valued} spinor~$\mathcal{F}$ and hence one can define a complex-valued function~$\mathcal{S}:=H_\mathcal{F}$ on~$G^\bullet\cup G^\circ$ by~\eqref{eq:def_of_H}. We interpret~$\mathcal{S}$ as an embedding of~$G$ into~$\mathbb{C}$ and call it \emph{s-embedding}. Note that, similarly to~\cite[Section~3.2]{russkikh}, the function~$\mathcal{S}$ can be also defined on~$z\in\diamondsuit(G)$ by
\begin{equation}
\label{eq:H_on_diamond}
\begin{array}{rcl}
\mathcal{S}(v^\bullet_p)-\mathcal{S}(z) & :=&\cos{\theta}\cdot \mathcal{F}(c_{p,0})\mathcal{F}(c_{p,1}),\\
\mathcal{S}(z)-\mathcal{S}(v^\circ_q) & := & \sin{\theta}\cdot \mathcal{F}(c_{0,q})\mathcal{F}(c_{1,q}),
\end{array}
\end{equation}
where the two corners~$c_{p,0}$ and~$c_{p,1}$ (resp., $c_{0,q}$ and~$c_{1,q}$) are chosen on the same sheet of~$\Upsilon^\times(G)$; these equations are consistent with~\eqref{eq:def_of_H} due to~\eqref{eq:propagation_on_corners}.

It is easy to see that~\eqref{eq:propagation_on_corners} also implies the identity
\begin{align}
\notag |\mathcal{S}(v^\bullet_0)-\mathcal{S}(v^\circ_0)|+|\mathcal{S}(v^\bullet_1)-\mathcal{S}(v^\circ_1)| & =|\mathcal{F}(c_{00})|^2+|\mathcal{F}(c_{11})|^2\\
\label{eq:tangential} = |\mathcal{F}(c_{01})|^2+|\mathcal{F}(c_{10})|^2 & = |\mathcal{S}(v^\bullet_0)-\mathcal{S}(v^\circ_1)|+|\mathcal{S}(v^\bullet_1)-\mathcal{S}(v^\circ_0)|,
\end{align}
which means that~$\mathcal{S}(v^\bullet_0)$, $\mathcal{S}(v^\circ_0)$, $\mathcal{S}(v^\bullet_1)$ and $\mathcal{S}(v^\circ_1)$ form a \emph{tangential} (though possibly non-convex) quadrilateral in the plane. In fact, one can easily see that~$\mathcal{S}(z)$, if defined according to~\eqref{eq:H_on_diamond}, is the center of the circle inscribed into this quadrilateral.

\begin{remark} (i) Let us emphasize that the parameters~$\theta_e$ in the s-holomorphicity condition~\eqref{eq:propagation_on_corners} are \emph{defined} as~$\theta_e:=2\arctan x_e$ and have no straightforward geometrical meaning similar to isoradial embeddings discussed above (though see~\eqref{eq:theta=}).

\noindent (ii) One can easily check that the isoradial embedding of the critical Z-invariant model is a particular case of the above construction in which~$\mathcal{F}(c)=\varsigma\delta^{1/2}\cdot \overline{\eta}_c$; note that, in this case, the Dirac spinor~\eqref{eq:Dirac_spinor} satisfies the propagation equation~\eqref{eq:propagation_on_corners}.
\end{remark}

A priori, there is no guarantee that the combinatorics of the s-embedding constructed above matches the one of~$\diamondsuit(G)$, considered as an abstract \emph{topological} (i.e., embedded into~$\mathbb{C}$ up to homotopies) graph: the images~$(\mathcal{S}(v^\bullet_0)\mathcal{S}(v^\circ_0)\mathcal{S}(v^\bullet_1)\mathcal{S}(v^\circ_1))$ of quads~$(v^\bullet_0 v^\circ_0 v^\bullet_1 v^\circ_1)$ might overlap. Below we assume that this does \emph{not} happen and, moreover, all~$(\mathcal{S}(v^\bullet_0)\mathcal{S}(v^\circ_0)\mathcal{S}(v^\bullet_1)\mathcal{S}(v^\circ_1))$ are nondegenerate and oriented counterclockwise  (except maybe the `root' one). We call such~$\mathcal{S}$ \emph{proper s-embeddings}.

We now introduce a set of geometric parameters characterizing an s-embedding~$\mathcal{S}$ up to translations and rotations. For a quad~$(v^\bullet_0v^\circ_0v^\bullet_1v^\circ_1)=z\in\diamondsuit(G)$, let $r_z$ denote the radius of the circle inscribed into its image~$(\mathcal{S}(v^\bullet_0)\mathcal{S}(v^\circ_0)\mathcal{S}(v^\bullet_1)\mathcal{S}(v^\circ_1))$ and let
\[
\phi_{zv^\bullet_p}:=\frac{1}{2}\,\mathrm{arg}\frac{\mathcal{S}(v^\circ_{1-p})-\mathcal{S}(v^\bullet_p)} {\mathcal{S}(v^\circ_{p})-\mathcal{S}(v^\bullet_p)},\quad \phi_{zv^\circ_q}:=\frac{1}{2}\,\mathrm{arg}\,\frac{\mathcal{S}(v^\bullet_{q})-\mathcal{S}(v^\circ_q)} {\mathcal{S}(v^\bullet_{1-q})-\mathcal{S}(v^\circ_q)}
\]
be the half-angles of~$(\mathcal{S}(v^\bullet_0)\mathcal{S}(v^\circ_0)\mathcal{S}(v^\bullet_1)\mathcal{S}(v^\circ_1))$, note that~$\phi_{zv^\bullet_0}+\phi_{zv^\circ_0}+\phi_{zv^\bullet_1}+\phi_{zv^\circ_1}=\pi$ and
$|\mathcal{S}(v)-\mathcal{S}(z)|=(\sin\phi_{zv})^{-1}r_z$ for each of the four vertices~$v=v^\bullet_0,v^\circ_0,v^\bullet_1,v^\circ_1$.
A straightforward computation shows that
\begin{equation}
\label{eq:theta=}
\tan\theta_{v_0^\bullet v_1^\bullet} = \left(\frac{\cot\phi_{zv^\circ_0}\!+\cot\phi_{zv^\circ_1}}{\cot\phi_{zv^\bullet_0}\!+\cot\phi_{zv^\bullet_1}}\right)^{\!\!1/2}\! = \, \left(\frac{\sin\phi_{zv^\bullet_0}\sin\phi_{zv^\bullet_1}}{\sin\phi_{zv^\circ_0}\sin\phi_{zv^\circ_1}}\right)^{\!\!1/2},
\end{equation}
where~$\theta_e=2\arctan x_e$ is the standard parametrization of the Ising model weights.

\begin{remark} It is worth mentioning that the construction of an s-embedding described above is \emph{revertible}. Namely, given a proper embedding~$\mathcal{S}$ of~$(G^\bullet,G^\circ)$ into the complex plane formed by non-degenerate tangential quads, one can define a complex-valued spinor~$\mathcal{F}(c):=(\mathcal{S}(v^\bullet(c))-\mathcal{S}(v^\circ(c)))^{1/2}$ on~$\Upsilon^\times(G)$ and deduce from~\eqref{eq:tangential} that~\eqref{eq:propagation_on_corners} holds true for \emph{some}~$\theta_e$ (which must then coincide with~\eqref{eq:theta=}). In other words, given~$\mathcal{S}$, one can find Ising weights~$x_e=\tan\frac{1}{2}\theta_e$ on~$G$ and a pair~$\mathcal{F}_1,\mathcal{F}_2$ of real-valued spinors satisfying~\eqref{eq:propagation_on_corners} with these~$\theta_e$ such that~$\mathcal{S}=\mathcal{H}_\mathcal{F}$.
\end{remark}

\subsection{S-subharmonicity} Miraculously enough, it turns out that the cornerstone Lemma~\ref{lemma:positivity} actually \emph{admits} a generalization to the setup described above, though not a straightforward one, cf. Remark~\ref{rem:no-go-positivity}.
Let~$H$ be a function defined in a vicinity of a given vertex~$v^\bullet\in G^\bullet$ or~$v^\circ\in G^\circ$. We define its \emph{s-Laplacian}~$\Delta_{\mathcal{S}}H$ as
\begin{align*}
[\Delta_{\mathcal{S}}H](v^\bullet):=\sum\nolimits_{v_1^\bullet\sim v^\bullet} a_{v^\bullet v^\bullet_1}(H(v_1^\bullet)\!-\!H(v^\bullet))+ \sum\nolimits_{v^\circ\sim v^\bullet} b_{v^\bullet v^\circ}(H(v^\circ)\!-\!H(v^\bullet)),\\
[\Delta_{\mathcal{S}}H](v^\circ):= \sum\nolimits_{v^\bullet\sim v^\circ} b_{v^\circ v^\bullet}(H(v^\bullet)\!-\!H(v^\circ))- \sum\nolimits_{v_1^\circ\sim v^\circ} a_{v^\circ v^\circ_1}(H(v_1^\circ)\!-\!H(v^\circ)),
\end{align*}
where, for each quad~$(v_0^\bullet v_0^\circ v_1^\bullet v_1^\circ)=z\in \diamondsuit(G)$, one has
\begin{equation}
\label{eq:a_def}
a_{v^\bullet_0v^\bullet_1}=a_{v^\bullet_1v^\bullet_0}:=r_z^{-1}\sin^2\theta_{v_0^\bullet v_1^\bullet},\qquad a_{v^\circ_0v^\circ_1}=a_{v^\circ_1v^\circ_0}:=r_z^{-1}\cos^2\theta_{v_0^\bullet v_1^\bullet},
\end{equation}
and, for each edge~$(v^\bullet v^\circ)=(v^\bullet_0(z)v^\circ_1(z))=(v^\bullet_0(z')v^\circ_0(z'))$ separating $z,z'\in\diamondsuit(G)$,
\begin{equation}
\label{eq:b_def}
b_{v^\bullet v^\circ} := a_{v^\circ v^\circ_0(z)}-\frac{r_z^{-1}\cot\phi_{zv^\bullet}}{\cot\phi_{zv^\bullet}\!+\cot\phi_{zv^\circ}}  + a_{v^\circ v^\circ_1(z')} - \frac{r_{z'}^{-1}\cot\phi_{z'v^\bullet}}{\cot\phi_{z'v^\bullet}\!+\cot\phi_{z'v^\circ}}
\end{equation}
and~$b_{v^\circ v^\bullet}:=b_{v^\bullet v^\circ}$, so that~$\Delta_\mathcal{S}=\Delta_\mathcal{S}^\top$ is symmetric (though not sign-definite).

\begin{remark} For isoradial embeddings of graphs carrying the critical Z-invariant Ising model one has~$a_{v^\bullet_0v^\bullet_1}=\delta^{-1}\tan\theta_{v^\bullet_0v^\bullet_1}$,~$a_{v^\circ_0v^\circ_1}=\delta^{-1}\cot\theta_{v^\bullet_0v^\bullet_1}$ and~$b_{v^\circ v^\bullet}=0$, thus~$\Delta_\mathcal{S}$ is simply the direct sum of the two \emph{signed} Laplacians~$\delta^{-1}\Delta^\bullet$ and~$-\delta^{-1}\Delta^\circ$.
\end{remark}

\begin{lemma} Let~$\mathcal{S}=H_{\mathcal{F}_1+i\mathcal{F}_2}$ be a proper s-embedding and a function~$H_F$ be obtained via~\eqref{eq:def_of_H} from a real-valued s-holomorphic spinor~$F$ defined on~$\Upsilon^\times(G)$ in a vicinity of a vertex~$v\in \Lambda(G)$. Then,~$[\Delta_\mathcal{S} H_F]\ge 0$. Moreover, $[\Delta_\mathcal{S}H_F]=0$ if and only if~$F$ is a linear combination of~$\mathcal{F}_1$ and~$\mathcal{F}_2$ on corners of~$G$ incident to~$v$.
\end{lemma}

\begin{proof} As in the isoradial case (see~\cite[Proposition~3.6(iii)]{chelkak-smirnov}), we are only able to check this by brute force though clearly a more conceptual explanation of this phenomenon must exist. E.g., if one numbers the quads~$z_1,...,z_n$ around \mbox{$v^\bullet_0\in G^\bullet$} counterclockwise and adopts the notation~$z_s=(v^\bullet_0 v^\circ_{s-1} v^\bullet_s v^\circ_s)$, $c_s=(v_0^\bullet v_s^\circ)$, and~$\mathcal{F}(c_s)=\alpha e^{i(\phi_1+...+\phi_s)}\rho_s$ with~$\alpha=\exp[\,i\!\arg\mathcal{F}(c_0)]$ and~$\rho_s=|\mathcal{F}(c_s)|>0$, 
then
\begin{align}
\label{eq:a_via_F}
a_{v^\bullet_0 v^\bullet_{s}}&=\frac{\sin\theta_{s}\tan\theta_{s}}{\rho_{s-1}\rho_s\sin\phi_s},\qquad
a_{v^\circ_{s-1}v^\circ_s}=\frac{\cos\theta_s}{\rho_{s-1}\rho_s\sin\phi_s},\\
\label{eq:b_via_F} b_{v^\bullet_0 v^\circ_s} &= 
\frac{\cos\theta_s}{\rho_{s-1}\rho_s\sin\phi_s}+\frac{\cos\theta_{s+1}}{\rho_s\rho_{s+1}\sin\phi_{s+1}}
- \frac{\sin(\phi_s\!+\!\phi_{s+1})}{\rho_s^2\sin\phi_s\sin\phi_{s+1}}\,,
\end{align}
and~$[\Delta_\mathcal{S}H](v^\bullet_0)$ turns out to be a non-negative quadratic form in the variables~$F(c_s)$:
\[
[\Delta_\mathcal{S} H_F](v^\bullet_0)=Q^{(n)}_{\phi_1;...;\phi_n}(\rho_0^{-1}F(c_0),...,\rho_{n-1}^{-1}F(c_{n-1}))\ge 0\,,
\]
see \cite[p.~543]{chelkak-smirnov} for the definition of~$Q^{(n)}_{\phi_1;...;\phi_n}$. The case~$v\in G^\circ$ is similar.
\end{proof}

\begin{definition}
\label{def:s-(sub)harm} We call a function defined on (a subset of)~$\Lambda(G)$ \emph{s-subharmonic} if the inequality~$\Delta_\mathcal{S}H\ge 0$ holds true pointwise and \emph{s-harmonic} if~$\Delta_\mathcal{S}H=0$ pointwise.
\end{definition}

\begin{remark}
Though this is not fully clear at the moment, we hope that, at least in some situations of interest (e.g., see Section~\ref{sub:doubly-periodic} or~\cite{lis-circle-patterns}), {s-subharmonic} functions~$H_F$ obtained via~\eqref{eq:def_of_H} are \emph{a priori} close to {s-harmonic} ones; recall that this is exactly the viewpoint developed in~\cite[Section~3]{chelkak-smirnov} for the critical Z-invariant model. Also, note that extending the domain of definition of~$H_F$ to~$\diamondsuit(G)$ similarly to~\eqref{eq:H_on_diamond}, one can easily see that thus obtained functions satisfy the maximum principle.
\end{remark}

\subsection{Factorization of the s-Laplacian}
An important feature of the isoradial setup is the following factorization of the direct \emph{sum} of~$\Delta^\bullet$ and~$\Delta^\circ$, see~\cite{kenyon-inv-02} or \cite{chelkak-smirnov-adv}:
\begin{equation}
\label{eq:iso-Delta-fact}
-\delta^{-1}(\Delta^\bullet\!+\!\Delta^\circ)=16\partial_\Lambda^{\,*}\!R\, \partial_\Lambda^{\phantom{*}}=16\overline{\partial}{}_\Lambda^{*}\!R\, \overline{\partial}{}_\Lambda^{\phantom{*}}\,,
\end{equation}
where the Cauchy-Riemann operator~$\partial_\Lambda$ is given by~\eqref{eq:def_d_iso} and~$R:=\mathrm{diag}\{r_z\}_{z\in\diamondsuit(G)}$; note that one has~$r_z=\frac{1}{4}\delta^{-1}|v^\bullet_1-v^\bullet_0||v^\circ_1-v^\circ_0|$ in this special case. Since the s-Laplacian~$\Delta_\mathcal{S}$ is not sign-definite (recall that, on rhombic lattices,~$\Delta_{\mathcal{S}}$ is the \emph{difference}~$\delta^{-1}(\Delta^\bullet-\Delta^\circ)$), the factorization~\eqref{eq:iso-Delta-fact} cannot be generalized directly. However, there exists a way to rewrite\begin{footnote}{For rhombic lattices, one has a very special intertwinning identity~$\overline{\partial}_\Lambda=U\partial_\Lambda(-\mathrm{Id}^\bullet\!+\mathrm{Id}^\circ)$.}\end{footnote} it in the full generality of s-embeddings:
\begin{equation}
\label{eq:Delta_S_factorization}
\Delta_\mathcal{S}=16\partial_\mathcal{S}^{\,*}U^{-1}\!R\,\overline{\partial}{}^{\phantom{*}}_\mathcal{S}= 16\overline{\partial}{}_\mathcal{S}^{\,*}\overline{U}^{-1}\!R\,\partial^{\phantom{*}}_\mathcal{S},
\end{equation}
where
\[
[\overline{\partial}_{\mathcal{S}}H](z) := \frac{\mu_z}{4}\biggl[\frac{H(v^\bullet_0)}{\mathcal{S}(v^\bullet_0)\!-\!\mathcal{S}(z)}+ \frac{H(v^\bullet_1)}{\mathcal{S}(v^\bullet_1)\!-\!\mathcal{S}(z)}- \frac{H(v^\circ_0)}{\mathcal{S}(v^\circ_0)\!-\!\mathcal{S}(z)} -\frac{H(v^\circ_1)}{\mathcal{S}(v^\circ_1)\!-\!\mathcal{S}(z)}\biggr],
\]
the prefactor~$\mu_z$ is chosen so that~$[\overline{\partial}_{\mathcal{S}}\,\overline{\mathcal{S}}](z)=1$, the operator
$\partial_\mathcal{S}$ is defined 
so that~$\overline{\partial_\mathcal{S}H}=\overline{\partial}_{\mathcal{S}}\overline{H}$, and~$U:=\mathrm{diag}(\mu_z)_{z\in\diamondsuit(G)}$.
Moreover, the following is fulfilled:
\begin{lemma}
\label{lemma:harm-conj}
A real-valued function $H_1$ is (locally) s-harmonic on~$\Lambda(G)$ if and only if there exists (locally and hence globally in the simply conected setup) another real-valued s-harmonic function~$H_2$ on~$\Lambda(G)$ such that~$\overline{\partial}_{\mathcal{S}}(H_1\!+\!i H_2)=0$ on~$\diamondsuit(G)$. In other words, s-harmonic functions are real parts of those lying in the kernel of~$\overline{\partial}_\mathcal{S}$.
\end{lemma}
Another straightforward computation shows that~$\overline{\partial}_\mathcal{S}H=0$ if~$H$ is a constant and that~$\overline{\partial}_\mathcal{S}\mathcal{S}=0$, which (together with the normalization~$\overline{\partial}_\mathcal{S}\overline{\mathcal{S}}=1$) establishes some link of the difference operator~$\overline{\mathcal{S}}$ with the standard complex structure on~$\mathbb{C}$. In addition, one can check that~$\overline{\partial}_\mathcal{S}L_\mathcal{S}=0$, where the real-valued function~$L_\mathcal{S}$ is defined on~$\Lambda(G)$, up to an additive constant, by~$L_\mathcal{S}(v^\bullet(c))-L_\mathcal{S}(v^\circ(c)):=|\mathcal{S}(c)|^2$.

\subsection{Doubly-periodic graphs} \label{sub:doubly-periodic}

We now briefly discuss s-embeddings of doubly-periodic graphs~$G$ carrying a \emph{critical} Ising model. It was shown in~\cite{cimasoni-duminil} that the criticality condition is equivalent to the existence of two periodic functions in the kernel of the Kac--Ward matrix~\eqref{eq:kac-ward}, which means (e.g., see~\cite{chelkak-cimasoni-kassel}) the existence of two linearly independent \emph{periodic} (real-valued) \emph{spinors}~$\mathcal{F}_1$, $\mathcal{F}_2$ on~$\Upsilon^\times(G)$. Thus, up to a global scaling and rotation (corresponding to the multiplication of~$\mathcal{F}$ by a constant), it remains to tune one \emph{complex-valued parameter~$\kappa$} in order to construct a periodic s-embedding~$\mathcal{S}=H_{\mathcal{F}_1+\kappa\mathcal{F}_2}$ of~$G$. The choice of~$\kappa$, in particular, corresponds to the choice of the \emph{conformal modulus~$\tau$} of (the image under~$\mathcal{S}$ of) a fundamental domain of~$G$. However, note that this dependence is not trivial: according to~\eqref{eq:def_of_H}, $\tau=\tau(\kappa)$ is the ratio of two quadratic polynomials constructed from~$\mathcal{F}_1$ and~$\mathcal{F}_2$.

Using~\eqref{eq:a_via_F} and~\eqref{eq:b_via_F}, one can check that the s-Laplacian~$\Delta_\mathcal{S}$ is essentially independent of the choice of~$\mathcal{F}$: changing~$\kappa$ results in the multiplication of all the coefficients of~$\Delta_\mathcal{S}$ by a constant. Nevertheless, the operator~$\overline{\partial}_\mathcal{S}$ is much more sensitive to this choice. We believe that the following picture is true\begin{footnote}{At the moment we do not have a full proof. However, let us mention that one can justify the missing ingredients for s-embeddings close enough to isoradial ones using continuity arguments.}\end{footnote}:
\begin{itemize}
\item Provided~$\mathcal{F}_1$ and~$\mathcal{F}_2$ are linearly independent, any choice of~$\kappa\not\in \mathbb{R}$ leads either to a proper s-embedding or to the conjugate of a proper s-embedding.
\item The kernel of~$\Delta_\mathcal{S}$ in the space of periodic functions is two-dimensional.
\item There exists a unique (up to conjugation) value~$\kappa=\kappa_L$ such that the function~$L_\mathcal{S}$ is periodic. For~$\kappa\ne\kappa_L$, the kernel of~$\overline{\partial}_{\mathcal{S}}$ in the space of periodic functions consists of constants only. For~$\kappa=\kappa_L$, it coincides with the kernel of~$\Delta_\mathcal{S}$ and is spanned by constants and~$L_\mathcal{S}$.
\item If~$\kappa=\kappa_L$, there exists a \emph{periodic} (and hence bounded) function~$\rho$ on~$\Lambda(G)$ such that the complex-valued function~$\mathcal{S}^2+\rho$ is s-harmonic (i.e., both its real and imaginary parts are s-harmonic, note that~$\rho=0$ for rhombic lattices).
\end{itemize}

\begin{remark}
The last claim reveals the rotational (and, eventually, conformal) symmetry of the critical Ising model on~$(G,x)$, which must show up if the conformal modulus~$\tau=\tau(\kappa_L)$ of the fundamental domain is tuned properly. Note that one should not expect that s-harmonic functions on a general s-embedding~$\mathcal{S}$ of~$G$ behave like continuous harmonic functions even on large scales, not to mention periodic or quasi-periodic fluctuations. Indeed, as mentioned above, the operator~$\Delta_\mathcal{S}$ is essentially independent of the choice of~$\kappa$ (and thus~$\tau(\kappa)$), hence in general one cannot hope for more than a skewed rotational symmetry despite of Lemma~\ref{lemma:harm-conj} and the basic properties~$\overline{\partial}_\mathcal{S} 1 =\overline{\partial}_\mathcal{S}\mathcal{S}=0$ of the Cauchy--Riemann type operator~$\overline{\partial}_\mathcal{S}$. Nevertheless, by analogy with usual discrete harmonic functions, one can hope that some form of an invariance principle for s-harmonic functions can be found.
\end{remark}

\section{Open questions}

Above, we already indicated several promising research directions basing on the analysis of s-holomorphic spinors, notably `CFT on the lattice level' projects (see Section~\ref{sub:lattice-CFT}) and the study of s-embeddings of planar graphs. Besides universality questions, one can apply them to \emph{random maps} (finite or infinite) carrying the Ising model, in an attempt to understand their conformal structure in the large size limit (one of the most straightforward setups is to interpret a random quandrangulation as~$\diamondsuit(G)$ and to work with the self-dual Ising weights~$x_\mathrm{sd}=\sqrt{2}-1$ on~$G^\bullet$,~$G^\circ$).

\smallskip

Another challenging research direction, which {cannot} be reduced to the analysis of s-holomorphic observables and requires some other techniques to be developed, is a better understanding of \emph{topological correlators}, cf.~\cite{ikhlef-jaspersen-saleur}. E.g., one can consider the FK-Ising model with $2k$ marked points on the boundary of~$\Omega^\delta$ and alternating wired/free/.../wired/free boundary conditions, cf.~\eqref{eq:crossing-proba}. There is the Catalan number~$C_k$ of possible patterns of interfaces matching these marked points and only~$2^{k-1}\ll C_k$ correlations of disorders on free arcs to handle. The situation with domain walls is even worse: no geometric information on those can be extracted directly \emph{in discrete}: already in the simplest possible setup with four marked points and~`$+/-/+/-$' boundary conditions, the only available way to prove the convergence of the crossing probabilities is first to prove the convergence of interfaces to hypergeometric SLEs and then to do computations \emph{in continuum}, see~\cite{izyurov-free}.

\begin{remark} In the context of double-dimer and CLE(4) loop ensembles, it was recently demonstrated by Dub\'edat~\cite{dubedat-isomonodromy} that topological correlators can be treated via tau-functions associated with SL(2)-representations of the fundamental group of a punctured domain~$\Omega^\delta$ (see also~\cite{basok-chelkak}). This remarkable development raises the following question: could one attack topological correlators corresponding to CLE($\kappa$), $\kappa\ne 4$, replacing SL(2) by relevant quantum groups? If so, could one use the critical Ising model, once again, as a laboratory to reveal and to analyze these structures in discrete? Note also that a detailed understanding of such topological correlators would also pave the way to a better understanding of the famous Coulomb gas approach to the critical lattice models; see~\cite{nienhuis} and \cite[Section 5.3]{smirnov-06icm}.
\end{remark}

As discussed above, both the convergence of critical Ising correlation functions (to CFT ones) and that of loop ensembles (to CLEs) as~$\Omega^\delta\to \Omega$ are now understood in detail. Moreover, a scaling limit~$\sigma_\Omega(z)$, $z\in\Omega$, of the {spin field}~$\{\sigma_u\}_{u\in V^\circ(\Omega^\delta)}$ viewed as a {random distribution} (generalized function) was constructed in the work of Camia, Garban and Newman~\cite{camia-garban-newman-I}. This leads to natural measurability questions:
e.g., is it true that~$\sigma_\Omega$ is (not) measurable with respect to the nested CLE(3) -- the limit of domain walls -- and vice versa? In fact, \cite{camia-garban-newman-I} ensures the measurability of~$\sigma_\Omega$ with respect to the limit of FK-Ising clusters -- CLE(16/3) -- but the latter contains more information, cf.~\cite{miller-sheffield-werner-cle-percolation}. Also, one can wonder whether it is possible to construct the \emph{energy density} correlation functions out of these CLEs (e.g., via some regularized `occupation density' of loops)? If so, one could then try to generalize such a construction to $\kappa\ne 3$ (note that the Ising spin field~$\sigma$ is more model-specific from the CLE perspective than the energy operator~$\varepsilon\sim\phi_{3,1}$, cf.~\cite[p.~319]{dotsenko-fateev-84}).

\smallskip

We conclude this essay on the \emph{critical} Ising model by a famous question on the \emph{supercritical} one: to prove that each fixed value~$x>x_{\mathrm{crit}}$ gives rise to the CLE(6) as the scaling limit of domain walls configurations. Note that a possible approach to this via the study of \emph{massive} theories was suggested in~\cite[Question~4.8]{makarov-smirnov}.

\end{document}